\documentclass[journal,comsoc]{IEEEtran}
\usepackage[utf8]{inputenc}
\usepackage[T1]{fontenc}

\ifCLASSINFOpdf
\else
\fi
\usepackage{amsmath,amsthm}
\usepackage[cmintegrals]{newtxmath}
\usepackage{mathtools}

\usepackage{url}

\usepackage[english]{babel}
\usepackage{newtxtext}
\usepackage{graphicx}
\usepackage{booktabs}
\usepackage{subcaption}
\usepackage{csquotes}
\usepackage[
    binary-units=true
]{siunitx}
\usepackage{color}
\usepackage[ruled,linesnumbered,norelsize,noend]{algorithm2e}

\usepackage{microtype}

\newtheorem{theorem}{Theorem}
\newtheorem{definition}{Definition}

\SetAlFnt{\footnotesize\sffamily}
\SetAlCapFnt{\small}
\SetAlCapNameFnt{\small}

\SetNlSty{textln}{}{}

\SetKwFunction{addVertexPerColor}{\textsf{addVertexPerColor}}

\SetKwFunction{computeShadowRating}{computeShadowRating}
\SetKw{True}{\textsf{true}}
\SetKw{False}{\textsf{false}}
\SetKwInOut{Input}{input}
\SetKwInOut{Output}{output}
\SetKwData{gch}{GCH}
\SetKwData{shadowRating}{shadowRating}

\DeclareMathOperator{\argmax}{arg\,max}

\newcommand{\markChanges}[1]{#1}

\begin{document}
%
\title{Dynamic QoS-Aware Traffic Planning for Time-Triggered Flows in the Real-time Data Plane}
%
%
%

\author{Jonathan Falk, Heiko Geppert, Frank Dürr, Sukanya Bhowmik, Kurt Rothermel
\thanks{This research was funded by the Deutsche Forschungsgemeinschaft (DFG, German Research Foundation) -- 285825138.}
\newline Institute for Parallel and Distributed Systems, University of Stuttgart, Germany\newline%
$\left\{\text{firstname.lastname}\right\}$@ipvs.uni-stuttgart.de%
}

%
%

\markboth{IEEE Transactions On Network and Service Management,~Vol.~19, No.~2, JUNE 2022
}%
{Dynamic QoS-Aware Traffic Planning}
%

\IEEEpubid{1932-4537~\copyright~2022 IEEE. Personal use is permitted, but republication/redistribution requires IEEE permission.
}


\maketitle

\begin{abstract}
	Many networked applications, e.g., in the domain of cyber-physical systems, require strict service guarantees for time-triggered traffic flows, usually in the form of jitter and latency bounds.
	It is a notoriously hard problem to compute a network-wide traffic plan, i.e., a set of routes and transmission schedules, that satisfies these requirements, and dynamic changes in the flow set add even more challenges.
	Existing traffic-planning methods are ill-suited for dynamic scenarios because they either suffer from high computational cost, can result in low network utilization, or provide no explicit guarantees when transitioning to a new traffic plan that incorporates new flows.
	
	Therefore, we present a novel approach for dynamic traffic planning of time-triggered flows.
	Our conflict-graph-based modeling of the traffic planning problem allows for the  reconfiguration of active flows to increase the network utilization, while also providing per-flow QoS guarantees during the transition to the new traffic plan.
	Additionally, we introduce a novel heuristic for computing the new traffic plans.
	Evaluations of our prototypical implementation show that we can efficiently compute new traffic plans in scenarios with hundreds of active flows for a wide range of settings.
\end{abstract}


%
\IEEEpeerreviewmaketitle

\section{Introduction}
Real-time communication with guarantees on delay and jitter has become an essential requirement for implementing time-sensitive networked systems from application domains such as manufacturing (Industrial Internet of Things), the automotive domain, or any kind of cyber-physical system, where physical processes are controlled through networked sensors, actuators, and controllers.
For a long time, networking technologies for such time-sensitive systems have followed another trajectory than \enquote{traditional} computer networks resulting in field-buses and other distinct, commercial networking standards that were tailored to the specific real-time requirements for so-called operational technology (OT).
However, to save cost and to benefit from the fast evolution of wide-spread networking technologies such as Ethernet, we observe a strong trend to merge these two networking domains.
This trend manifests itself in activities of standardization bodies, namely, the IETF Deterministic Networking (DetNet)~\cite{rfc8557} initiative or IEEE Time-Sensitive Networking (TSN).

For example, TSN has brought several extensions that equip IEEE Std. 802.3 (Ethernet) LANs with mechanisms allowing for deterministic bounds on network delay and jitter.
In particular, the Time-Aware Shaper (TAS)  (cf. IEEE Std. 802.1Q~\cite{ieeecomputersocietyIEEEStandardLocal2018}) provides a mechanism to schedule the release times of frames in the bridges in the network.
In combination with explicit routing, this allows to enforce a network-wide traffic plan consisting of routes and schedules for each individual flow that provides hard real-time guarantees for time-triggered flows.
Here, we use the term time-triggered flow as the umbrella term for isochronous traffic or cyclic-synchronous traffic, which are traffic types originating in field-level OT with cycle times and jitter and latency bounds in the range of micro-seconds.

\IEEEpubidadjcol

Computing such a traffic plan is well-known to be a (NP-)hard problem, for instance, related to the Job Shop Scheduling Problem \cite{durr_no-wait_2016-1}.
So far, the traffic planning problem has predominantly been addressed for static, a priori known sets of flows with a focus on scalability to support larger scenarios.

However, a second general trend which is postulated both in industry white-papers~\cite{noauthor_hello_2016,noauthor_hello_2016-1} and research articles~\cite{raagaard_runtime_2017-1,danielisRealTimeCapable2018,prinz_dynamic_2018} is a general movement away from static environments towards dynamic reconfigurable scenarios, e.g., when attaching or physically moving a sensor or machine in a smart factory.
The term \emph{plug-and-produce} is frequently used in the context of Industry 4.0 to describe the need for flexible production facilities where devices can be added quickly and configured automatically, which also includes the adaption of network schedules and routes to integrate new devices into the network.

Software-defined Networking (SDN) introduced the fundamental concept of a logically centralized network controller with a global view onto the network, in particular, to dynamically adapt routes.
TSN similarly provides the concept of a centralized network controller (CNC) and associated management abstractions to automate network configurations including schedules.
Thus, the basic primitives for dynamically changing routes and schedules of nodes in a softwarized time-sensitive network exist.
However, algorithms for \emph{traffic planning} that such a controller has to provide are still subject to research.

This applies especially to traffic planning that supports dynamic scenarios where flows can be added dynamically---which is more than just \enquote{static planning done fast}.
In addition to the general challenges of network updates~\cite{reitblatt_abstractions_2012}, we also have to consider the quality-of-service (QoS) requirements of the time-triggered flows during the transitions between traffic plans.
For example, if we naively execute static traffic planning again from scratch each time new flows are to be added, this provides no control over the service degradation experienced by an individual flow, and usually we cannot even guarantee that the new plan still includes all former active flows.
Worse yet, we might even have to suspend the emission of new packets for some time to ensure that there is no interference induced by old packets---resulting in a \enquote{stop-and-go}-reconfiguration, or, if technically possible, we have to drop old packets.
All of this boils down to the insight that dynamic traffic planning requires taking the current traffic flows into account when computing a new traffic plan.

In early approaches~\cite{raagaard_runtime_2017-1,nayakIncrementalFlowScheduling2018} this is addressed by what amounts to defensive planning:
Defensive planning does not change the configuration of active flows in order to add new flows.
In other words, defensive planning takes the configuration of active flows as fixed and uses the remaining network resources for routing and scheduling new flows.
While defensive planning never affects the quality of service (QoS) of active flows, it might utilize network resources in a sub-optimal way.
For example, defensive planning might result in scenarios where new flows have to be rejected simply because the active flows have \enquote{fragmented} the remaining network resources, and we cannot \enquote{correct} the past scheduling and routing decisions.

Obviously, this issue can be mitigated by what amounts to offensive planning where (some) active flows can be reconfigured when adding new flows.
However, this introduces additional challenges:
While offensive planning allows for better utilization of network resources, the transitory effects of reconfigurations possibly introduce short-term QoS degradation~\cite{li_enhanced_2019}, and the transition phase has to be temporally bounded to guarantee deterministic communication in the long run~\cite{balasubramanian_sdn_2021}.
This means, offensive planning makes only sense if we can control both, the degree and duration of a QoS degradation, and it requires applications that can tolerate controlled fluctuations of QoS.
Such applications, for instance, correspond to the cyclic-synchronous traffic pattern where the emphasis is on latency and some jitter is acceptable, cf. the upcoming IEEE/IEC 60802 TSN profiles~\cite{noauthor_iecieee_nodate}.

In the literature, reconfigurations of time-triggered flows are addressed in~\cite{li_enhanced_2019,pang_flow_2021}.
The approach in~\cite{li_enhanced_2019} considers the problem how to execute an update from one given, current traffic plan to another given, new traffic plan such that the number of lost packets is minimized.
This is later extended with a method~\cite{pang_flow_2021} for the computation of the new traffic plan, more specifically, a new schedule.
The problem of service degradation is only addressed on the level of packet losses, i.e., the computation of a new traffic plan is subject to constraints which ensure that a loss-free packet update is possible.

In contrast, our traffic planning approach considers both, the temporal (scheduling) and spatial (routing) aspects.
We also address the challenges of offensive planning for cyclic-synchronous traffic flows by providing a method to satisfy per-flow QoS degradation requirements in terms of packet timings (not just no packet drops).
In detail, the contributions of this paper are twofold:
\begin{itemize}
	\item
    We present a novel approach for dynamic traffic planning for isochronous and/or cyclic-synchronous traffic flows using the zero-queuing principle~\cite{durr_no-wait_2016-1,perryFastpassCentralizedZeroqueue2014}.
	Our approach supports offensive traffic planning where the transitions from the old traffic plan to the new traffic plan, i.e., a traffic plan update, do not require artificial pauses of sender nodes or dropping packets, and our approach supports limiting the QoS degradation during the update according to per-flow user-specified bounds.
    This allows for a mixed-operation mode where a subset of active flows (e.g., of jitter-resilient applications) can be reconfigured, while the remaining active flows (e.g., of jitter-sensitive applications) are exempt from reconfigurations.
    
	\item
	We use a conflict-graph-based method \cite{falkTimeTriggeredTrafficPlanning2020} for the traffic planning.
	Conflict-graph-based modeling is well-suited for dynamic traffic planning because the conflict graph embeds a large portion of the knowledge about the solution space from the previous traffic plan, and thus reduces the effort required to compute the new traffic plan.
	Computing a new traffic plan requires to solve a variant of the independent colorful vertex set problem with weighted colors.
    We present a novel heuristic for the underlying optimization problem which we have to solve to obtain a new traffic plan.
	Our evaluations show that our novel heuristic can solve this problem for hundreds of flows in a fraction of the time required by a state-of-the-art optimizer.
\end{itemize}
Following the related work discussion in Sec.~\ref{sec:related_work}, we introduce our system model, modeling concepts, and problem statement in Sec.~\ref{sec:system_model}-\ref{sec:concepts_problem}.
We describe the traffic planning approach in Sec.~\ref{sec:constructingcg}--\ref{sec:qos}, evaluate our approach in Sec.~\ref{sec:eval} and conclude the paper in Sec.~\ref{sec:conclusion_outlook}.


\section{Related Work}
\label{sec:related_work}

In this paper, we address dynamic traffic planning for isochronous or cyclic-synchronous traffic for packet-switched, meshed networks with hard real-time requirements in terms of end-to-end latency and jitter.
In particular, we focus on time-triggered approaches that rely on explicitly scheduled transmissions throughout the network.
This means, approaches that use a combination of traffic shaping and admission control~\cite{guckDetServNetworkModels2017} are considered out of scope, since they are more suited for jitter-tolerant applications with asynchronous traffic patterns.
Work that is primarily concerned with aspects of interfaces or system architectures~\cite{prinz_dynamic_2018,gutierrezSelfconfigurationIEEE8022017,nasrallah_reconfiguration_2019,gerhard_software-defined_2019} addresses a complimentary problem, and usually relies on traffic planning functionality.


TSN and similar converged networking technologies such as TTEthernet (SAE AS6802), or Time-Sensitive Software Defined Networks (TSSDN) have sparked many traffic planning approaches for scheduling time-triggered isochronous and cyclic-synchronous traffic in static scenarios.
These can be roughly classified into two categories: 1) approaches that rely on generic constraint programming frameworks such as Integer Linear Programming (ILP) \cite{nayak_time-sensitive_2016-1,schweissguth_ilp-based_2017,schweissguth_ilp-based_2017,jonathanfalkExploringPracticalLimitations2018,atallahRoutingSchedulingTimeTriggered2019} or Satisfiability Modulo Theories (SMT) \cite{steiner_evaluation_2010,oliverIEEE8021Qbv2018,steinerTrafficPlanningTimeSensitive2018}, and  2) approaches that use custom heuristics \cite{durr_no-wait_2016-1,pahlevanHeuristicListScheduler2019,syedEfficientOfflineScheduling2019a} or a combination of both \cite{falkTimeTriggeredTrafficPlanning2020}.
Out-of-the-box, most static planning approaches cannot be used for dynamic traffic planning, since---by nature---they do not account for old packets from a previous traffic plan which results in the timing violations during the transition from one traffic plan to another.

In principle, incremental approaches~\cite{raagaard_runtime_2017-1,nayakIncrementalFlowScheduling2018} can be used for static scenarios and defensive planning in dynamic scenarios:
\enquote{Freezing} the configurations of active flows ensures that \enquote{nothing bad} happens when new flows are added, because the configurations of active flows from each previous step remain immutable for the future.
However, defensive traffic planning prohibits adjustments of any past decision that may turn out to have been sub-optimal with respect to schedulability when facing new flows.
This can result in low network utilization if a flow that has been added early on obstructs new flows.

This could be overcome by offensive planning that allows for the reconfiguration of active flows.
Reconfigurations of time-triggered flows are addressed in~\cite{li_enhanced_2019,pang_flow_2021}.
An enhancement for two-phase network updates~\cite{reitblatt_abstractions_2012} with schedule-aware per-flow update times was presented in \cite{li_enhanced_2019}.
The proposed update mechanism aims to reduce packet loss---caused by missed timings due to the reconfiguration---when transitioning from the current traffic plan to a given, new traffic plan.
While~\cite{li_enhanced_2019} considers traffic planning as external step, \cite{pang_flow_2021} presents an ILP-based approach for the computation/modification of schedules which eliminates the possibility of missed timings altogether.
Two algorithms are described in~\cite{pang_flow_2021} that both allow for packet-drop-free reconfigurations of active flows (offensive planning), by either computing the new schedule directly, or modifying an otherwise obtained, given schedule.

In contrast, we consider scheduling \emph{and routing} when computing the new traffic plan.
Our modeling of the problem and our novel heuristic are more efficient than constraint-based approaches because we can reuse most of the state which encodes the solution space from the previous step.
From a performance perspective, our approach therefore is a step towards true \enquote{online}-reconfiguration capabilities in terms of runtimes.

So far, the existing approaches provide service guarantees during a reconfiguration either by prohibiting reconfigurations (defensive planning) or only in terms of packet drops rather than timing guarantees.
In contrast, our approach is more powerful and allows for incorporating application-specific \emph{per-flow} bounds on the (temporary) QoS degradation in terms of jitter and packet reorderings, i.e., our approach enables offensive traffic planning for time-triggered flows with heterogeneous QoS requirements, e.g., wrt. jitter-tolerances.

\section{System Model}
\label{sec:system_model}
We consider traffic flows in packet-switched data networks.
A traffic flow is a sequence of packets that is transmitted from a source node to a destination node through the network.
The network consists of infrastructure nodes (switches, routers), source nodes, destination nodes, and a logically centralized controller.
Infrastructure nodes forward data packets via point-to-point links which connect nodes.
Source node and destination node are the origin and target of the packets of a flow, respectively.
Packets are sent and forwarded according to the global \emph{traffic plan}.
The traffic plan is computed and disseminated by the planner, which is situated at the logically centralized controller and has global view onto the network.

\subsection{Flow Dynamics}
Whenever a set of new flows is to be added to the network, the planner computes a new traffic plan.
Therefore, over time the planner computes a sequence of traffic plans, where a newly computed traffic plan becomes the current plan, which replaces the old one, cf. Fig.~\ref{fig:flow_dynamics}.
A traffic plan, say \(p\), defines for the set of flows admitted to the network when and along which routes the packets of the active flows are transmitted.
\begin{figure}
	\centering
	\includegraphics[clip,width=1\linewidth]{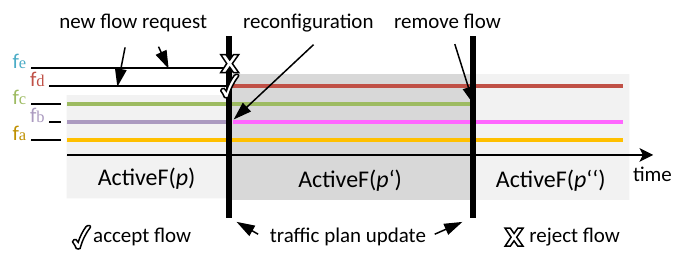}
	\caption{
            The planner computes a sequence of traffic plans.
            During the first traffic plan update, out of the two flow requests (ReqF(\(p'\))=\(\{f_d,f_e\}\)), only \(f_d\) can be admitted (which requires a reconfiguration of \(f_b\)) and \(f_d\) is added to the set of active flows ActiveF(\(p'\)), while \(f_e\) could not be scheduled and is rejected.
            The second update removes \(f_c\), which is not needed by the application anymore.}
	\label{fig:flow_dynamics}
\end{figure}

We will use ActiveF(\(p\)) to denote the set of active, i.e., admitted flows of plan \(p\).
For example, in Fig.~\ref{fig:flow_dynamics}, \(\operatorname{ActiveF}(p)=\{f_a,f_b,f_c\}\), \(\operatorname{ActiveF}(p')=\operatorname{ActiveF}(p) \cup \{f_d\}\), etc.
Active flows can also be removed, e.g., \(f_c\) in Fig.~\ref{fig:flow_dynamics}, which means that the source node stops sending packets and the planner is notified that it can reassign the resources previously reserved for that flow when computing the next traffic plan update.
Removing active flows is trivial from the perspective of the planner, cf. Sec.~\ref{sec:rejects_removals}.
In the following, we therefore concentrate on the case that the planner only attempts to add new flows to the network.

Let ReqF(\(p\)) be the set of flows that the applications request to be added to the network, i.e., ReqF(\(p\)) are the requests received by the planner while \(p\) is valid.
By processing this request, the planner generates a new traffic plan \(p'\) with \(\operatorname{ActiveF}(p')\setminus \operatorname{ActiveF}(p) \subseteq \operatorname{ReqF}(p)\), where \(p\) is the plan preceding \(p'\).
In the following, we denote the current traffic plan by \(p\) (which becomes the old traffic plan after the update), and the new traffic plan by \(p'\).
For example, in Fig.~\ref{fig:flow_dynamics}, ReqF(\(p\))=\(\{f_d, f_e\}\) and ReqF(\(p'\))=\(\emptyset\).

In general, it can happen that a flow in ReqF(\(p\)) may not be admitted (cf. \(f_e\) in Fig.~\ref{fig:flow_dynamics}) due to network resource limitations.
However, once a flow \(f\) has been admitted into the set of active flows, the end-to-end deadline constraints and jitter-bounds, i.e., real-time constraints, are guaranteed until the end of its lifetime.
This explicitly includes all future traffic plans and possible reconfigurations of \(f\).
For example, \(f_b\) in Fig.~\ref{fig:flow_dynamics} is an active flow whose lifetime spans over the two depicted traffic plan updates, the first one resulting in a reconfiguration of \(f_b\): during the whole time, the source node of \(f_b\) can send packets which are guaranteed to arrive at the destination node in time.

We focus specifically on the problem of \emph{computing} the traffic plan.
However, ultimately our traffic planning approach is but one building block of the controller, which likely will need to incorporate additional application-specific logic and interfaces.
For instance, the event that triggers the computation of a new traffic plan might depend on the application:
The traffic plan update could be initiated by applications issuing a flow request or a batch of flow requests.
Alternatively, the planner could collect incoming flow requests and decide itself when to compute a new traffic plan, e.g., periodically, or when the number of flow requests reaches a certain threshold.
Consequently, we do not give any guarantees on the response time for flow requests.
In other words, our goal is a \emph{real-time data plane}, not a real-time  control plane, although we strive for a fast processing of requests through efficient planning.
We will show later that we can parameterize our approach with regard to the trade-off between runtime and schedulability.

We assume that we do not have any prior knowledge as to how the flow set changes over time, else we could pre-compute a sequence of traffic plans.
In the following, the planner processes a flow request only once, i.e., either admits the new flow, or rejects it, but our approach allows for implementing other strategies, for example, retries of flow requests, too.


\subsection{Time-Triggered Traffic Flow}
\label{sec:deterministic_flows}
The notion of a traffic flow represents a directed communication channel from one source node to one destination node.
In the following, we focus on time-triggered flows that carry time-sensitive, important information requiring deterministic QoS and no packet drops under normal operation conditions, i.e., in absence of external failures such as power-loss, physical disturbances, etc.

Starting from the time of its activation, the source node of an active flow emits one packet per transmission cycle.
A transmission cycle is an interval of length \(t_\text{cycle}\).
Transmission cycles start at \(k\cdot t_\text{cycle} +t_0, k \in \mathbb{N}\) where \(t_0\) is a common reference point in time.
Source, destination, packet-size, and \(t_\text{cycle}\), and end-to-end delay are flow \emph{parameters} and do not change during a flow's lifetime.

When an application requests a flow, it specifies the parameters in the request for the planner.
Conversely, the planner assigns a \emph{route} which connects source node and destination node, and the \emph{phase} \(\phi\) to each active flow when computing the traffic plan.
The value of \(\phi\) is used by the source node to schedules its packet transmissions at \(k\cdot t_\text{cycle} + \phi +t_0, k \in \mathbb{N}\).
The valid range of \(\phi\) is restricted to \([0,t_\text{cycle}-t_\text{trans}]\) with \(t_\text{trans}=\frac{\text{packet size}}{\text{bandwidth}}\) being the transmission delay.
Thereby, the transmission of a packet on the outgoing link at a source node is confined within the transmission cycle boundaries.
Once the source node has sent a packet, the packet is forwarded by infrastructure nodes to the destination node.
Unless ambiguous, we omit the flow indices of flow properties, i.e., we use \(t_\text{cycle}\) instead of \(t_\text{cycle}(f)\).


\subsubsection{Flow Configuration}
A flow configuration is one assignment of values to the phase and route of a flow.
Since a time-continuous \(\phi\) yields infinitely many flow configurations, we make some practical assumptions:
As proposed in~\cite{falkTimeTriggeredTrafficPlanning2020}, we use discrete time with appropriately mapped granularity, e.g., with \SI{1}{\micro\second} resolution.

Secondly, we also restrict the routing options for each flow to a set of candidate paths.
\markChanges{
Each candidate path must be loop-free.
Additionally, we enforce the end-to-end delay constraint of a flow by capping the length of its candidate paths.
In Sec.~\ref{sec:zero-queuing}, we discuss how zero-queuing provides a simple mapping of the end-to-end delay bound \(t_\text{e2e}\) to the path length.
}
Depending on the scenario, e.g., a highly meshed network and \enquote{large} end-to-end delay bounds, we might still end up with impractical many candidate paths which are \enquote{short enough} to fall below the end-to-end delay induced path length limitation.
Therefore, for each flow, an absolute upper-bound on the number of candidate paths (\(n_\text{path}\)) can be specified to further restrict the routing options.

Upon receiving a request to add a flow, the planner then \emph{once} computes a set of up to \(n_\text{path}\) candidate paths for that flow, and assigns a path identifier \(\pi\) to each candidate path of the flow.
In our case, the planner uses Yen's k-shortest path algorithm~\cite{yenFindingShortestLoopless1971} to compute the candidate paths.

Hence, we can define a flow configuration with a tuple of integers from a finite set.
\begin{definition}[Flow Configuration]
	The tuple \((f,\phi,\pi)\) represents a flow configuration.
\end{definition}

Remember, offensive planning allows for the reconfiguration of active flows while a new plan is established.
Consequently, a flow may have different configurations over time.
The notion config(\(f,p\)) identifies the configuration of flow \(f\) while \(p\) is valid.
If the planner decides to reconfigure \(f\) in a succeeding plan \(p'\), then \(\operatorname{config}(f,p) \neq \operatorname{config}(f,p')\).
We say that each reconfiguration of a flow generates a new version of this flow, where each version exists as long as the corresponding plan is valid.
To be able to distinguish the various flow versions, we will use the notion \(\langle f,p \rangle\), where flow \(f\) is associated with plan \(p\), i.e., config(\(f,p\)) is the configuration of \(\langle f,p \rangle\).

\subsection{Node Capabilities}
\label{sec:netnode_capabilities}
We consider infrastructure nodes with store-and-forward behavior and FIFO queuing semantics per output port.
Infrastructure nodes provide facilities to forward packets along \emph{explicitly} defined routes, e.g., via SDN flow tables, VLAN tagging, explicit path control in IEEE Ethernet networks~\cite{isoISOIECIEEE2017}, etc.

Source nodes support time-triggered packet scheduling, i.e., it is possible to explicitly specify the time of a packet transmission.
Such functionality can be implemented in hardware, e.g., the LaunchTime-feature in  commodity NICs such as Intel's I210 controllers~\cite{vikramdadwalintelcorporationAdoptingTimeSensitiveNetworking}, or in software, e.g., TAPRIO or Earliest TxTime First (ETF) queuing discipline in Linux.
Consequently, all nodes need synchronized clocks, e.g., using protocols such as Precision Time Protocol (PTP).
Support for multiple traffic classes, and time-triggered packet scheduling (e.g., TAS, TTEthernet, etc.) in the infrastructure nodes can be used to protect the scheduled transmission windows from cross-traffic due to unscheduled flows.
\markChanges{
With time-triggered scheduling support in infrastructure nodes, unscheduled traffic, e.g., best-effort traffic, can be forwarded \enquote{outside} of these scheduled transmission windows.
In practice, this requires additional guard bands to ensure that the transmission of non-preemptible best-effort packets (or framelets) cannot continue into the transmission windows for time-triggered traffic.\footnote{
In the simplest case, we can create such guard bands by blocking unscheduled traffic for an interval of the length of MTU transmission before a transmission window for scheduled traffic.}
}
The schedules and reserved transmission windows for mechanisms such as the TAS can effectively be derived from the phase \(\phi\) of the corresponding configuration~\cite{falkTimeTriggeredTrafficPlanning2020}.

Before a node sends or forwards packets of a particular flow \(f\), the controller propagates the necessary flow information corresponding to config(\(f,p\)) to the respective nodes in the network.
What this entails is technology-specific, but includes routing entries and packet transmission schedules.
The flow information in the nodes, in particular routes and schedules, can be updated at runtime by the controller.
Flow information updates shall not affect in-flight packets to prevent network update issues such as black-holing.
This means, infrastructure nodes temporarily may have multiple sets of flow information for each version of a reconfigured flow and deliver each packet according to the flow configuration that was used by the source node when sending the packet.
This can be achieved, e.g., by tagging packets with additional metadata, address schemes, etc.~\cite{reitblatt_abstractions_2012}

The controller lazily purges obsolete flow information sets from the nodes once all old packets have been received by the destination nodes.
In this context, obsolete flow information sets refers to flows that either have been removed from ActiveF(\(p'\)), or have been reconfigured.
In the latter case, while the flow itself remains active, the flow information sets for old flow versions \(\langle f,p \rangle\) are not used anymore.
Practically, to perform this \enquote{garbage collection}, the controller has to track the flow information sets associated with each traffic plan, and then has to delete the corresponding routing entries, clear now unused scheduling entries, etc.

\markChanges{
In principle, an out-of-band mechanism (e.g., via an additional, dedicated management network) may be employed to perform the modifications of the flow information sets in the nodes.
However, we expect it to be more realistic to modify flow information at the network nodes using logical, in-band control channels, that run over the same network as the application data.
In Sec.~\ref{sec:installing_plan}, where we detail the installation of a traffic plan, we discuss the requirements of such control channels, and sketch how to implement them in-band.
}

We consider a network where all links operate with the same transmission speed.
We denote the processing delay of an infrastructure node by \(t_\text{proc}\) and the propagation delay on a link by \(t_\text{prop}\).
These values relate to hardware properties (switching fabric, cable length, etc.) and are therefore considered constants.
While it is not necessary for our approach that all nodes and links induce the same processing delay and propagation delay, respectively, we omit node and link identifiers for \(t_\text{proc}\) and \(t_\text{prop}\) in the following, i.e., assume they have the same values network-wide.

\subsection{Zero-Queuing}
\label{sec:zero-queuing}
Our approach uses the zero-queuing principle~\cite{durr_no-wait_2016-1,perryFastpassCentralizedZeroqueue2014} where packets must not be buffered in the network.
Instead, packets always enter empty queues at each network element and are immediately forwarded to the next hop.
Zero-queuing can be achieved by a global coordination  (here: in the form of the traffic plan) of packet transmissions.
If two flows violate the zero-queuing constraint, we say they \emph{interfere}.

Fig.~\ref{fig:deterministic_packets} depicts a packet sequence of a time-triggered flow at different points in the network with zero-queuing.
The source node sends one packet per cycle on its outgoing link starting at \(t_0 + \phi + k \cdot t_{cycle}, k \in \mathbb{N}\).
At each hop, packets accumulate only the (inevitable) delay which consists of the time between receiving a packet and completing the forwarding of the packet on the next link.
The per-hop delay \(t_\text{perhop}=t_\text{trans} + t_\text{prop}+t_\text{proc}\) can be computed from \(t_\text{trans}\) and the network properties \(t_\text{prop}\) and \(t_\text{proc}\).
Thus, zero-queuing enables transmissions with bounded end-to-end delay 
\(
t_\text{e2e} =  t_\text{trans} + t_\text{prop} + l \cdot t_\text{perhop} + t_\text{src} + t_\text{dst}
\)
that depends on the number of hops \(l\) in-between source and destination (i.e., \(l\) is the number of infrastructure nodes on the route), and processing delays of source and destination node (denoted by \(t_\text{src}\) and \(t_\text{dst}\)), respectively.
\begin{figure}
	\centering
	\includegraphics[
	trim={.6cm .6cm .6cm .6cm},clip,
	width=0.9\linewidth]{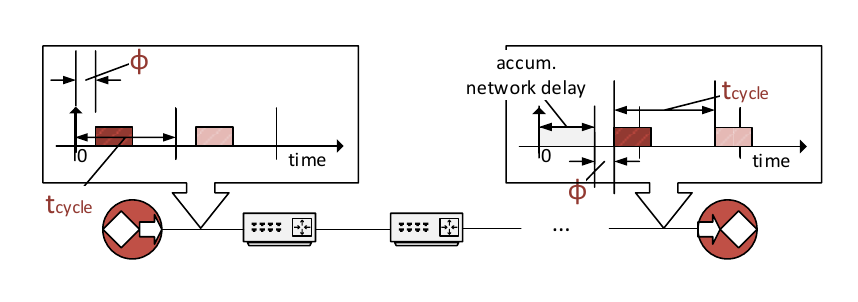}
	\caption{Zero-queuing: Packets accumulate deterministic delay on the route to the destination.}
	\label{fig:deterministic_packets}
\end{figure}

\section{Concepts and Problem Statement}
\label{sec:concepts_problem}
We use a conflict-graph-based approach to compute the new traffic plan.
Conflict-graph-based modeling for static traffic planning was previously introduced by us in~\cite{falkTimeTriggeredTrafficPlanning2020}.
From a birds-eye view, conflict-graph-based modeling represents flow configurations and their relations by vertices and edges in a so-called conflict graph.
A traffic plan can then be obtained from the conflict graph by searching for a set of conflict-free configurations.

Before we formally state our problem, we explain how these conflict graph concepts relate to our planning problem and extend these concepts for dynamic scenarios.

\subsection{Configuration Conflict}
Intuitively, we have a conflict between two configurations if packets would have to be transmitted \emph{simultaneously} on the same link to arrive at their respective destination without queuing.
Since packet transmission are serialized at each output port, we define configuration conflicts as violations of the zero-queuing constraint.
\begin{definition}[Conflict]
	Let \(c_1\) and \(c_2\) be two different configurations.
	The two configurations are in conflict if any packet sent with \(c_1\) would be buffered due to a packet sent with \(c_2\), or vice versa.
\end{definition}
\begin{figure}
	\centering
	\includegraphics[
	trim={.6cm .6cm .6cm .6cm},
    clip,
	width=.9\linewidth
	]{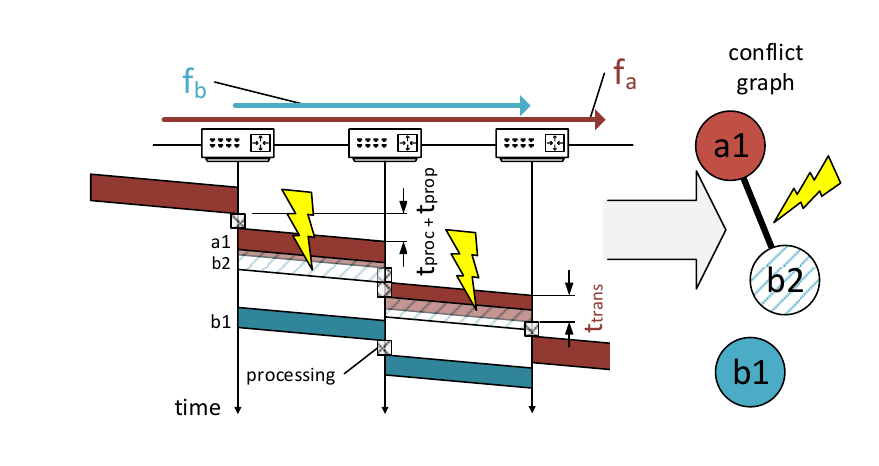}
	\caption{
        \markChanges{
            Example: The configuration \(a1\) for \(f_a\) conflicts with the configuration \(b2\) of \(f_b\) because \(f_b\)'s packets are scheduled too early, resulting in buffering.
            Both flows traverse all links without buffering if we select \(b1\) (which does not conflict with \(a1\)) for \(f_b\) instead.
        }
    }
	\label{fig:zeroqueuing}
\end{figure}
This is illustrated in Fig.~\ref{fig:zeroqueuing}.
There is a conflict between configuration \(a1\) for flow \(f_a\) and configuration \(b2\) for flow \(f_b\) since the transmission at the source of \(f_b\) is scheduled too early, and packets sent with \(b_2\) would have to be buffered until the transmission of packets from \(f_a\) is finished---which violates the zero-queuing principle.
In contrast, configuration \(a1\) and \(b1\) are conflict-free due to the increased phase for packets of \(f_b\).

The zero-queuing constraint allows us to conflict-check pairs of configurations independently in parallel.
To compute whether two configurations conflict, we have to check if the scheduled packet transmissions are always temporally isolated on common links.
Due to the cyclic property of the packet transmissions, we only check an interval of the length of a hyper-cycle of the two flows, i.e., the least common multiple of their \(t_\text{cycle}\)-values.
However, packets may take longer than the (hyper-)cycle to reach the destination.
Therefore, we use modulo arithmetic to \enquote{fold} transmissions crossing the cycle bounds back to the start of the interval.

The effort to check whether two configurations conflict in principle scales with the length of the candidate paths and the relative least common multiple (lcm) of the cycle times of the two flows (not the length of the global hyper-cycle which is the lcm of all cycle times).
Since for each pair of flows \(f_a,f_b\) we expect to encounter the same candidate-path combination \(\pi_a,\pi_b\) for many phase-values \(\phi_a,\phi_b\), we can use caching to store the common links for the candidate-path combinations.
In practice, the effort therefore depends on the shared sub-paths of the candidate paths and the lcm of the cycle times.

\subsection{Conflict Graph}
We encode the relations and constraints between the configurations in a conflict graph.
A flow configuration is represented by a vertex in the conflict graph, and conflicts between configurations of \emph{different} flows are represented by edges in the conflict graph.
For example, in Fig.~\ref{fig:zeroqueuing} there is an edge between configuration \textit{a1} and configuration \textit{b2}.

A new version of the conflict graph is generated by modifying the old one whenever a new traffic plan is to be computed.
As will be shown below, a new version of the conflict graph is generated by adding/removing configurations to the previous version and/or locking/unlocking configurations.
In the following, \(\mathbf{G}(p)\) identifies the version of the conflict graph that is associated with plan \(p\), i.e., \(\mathbf{G}(p)\) has been used to compute \(p\).

The conflict graph includes for each flow a set of potential configurations, which we will call the flow's candidate configurations, or candidates for short.
We use cand(\(f,p\)) to denote the candidate set of flow \(f\) in \(\mathbf{G}(p)\).
Given that---depending, e.g., on the granularity of the time-discretization---there potentially might exist millions or more configurations for each flow, we discuss our strategy for a non-exhaustive conflict graph construction in Sec.~\ref{sec:constructingcg:adding_candidates} where cand(\(f,p\)) need not contain every single existing configurations of \(f\).

We can interpret \(\mathbf{G}(p)\) as a colored graph.
\begin{definition}[Conflict Graph]
	The conflict graph \(\mathbf{G}(p)\) is an undirected, vertex-colored graph
	where all configuration vertices in cand(\(f,p\)) are colored by \(f\).
	Two configuration vertices in \(\mathbf{G}(p)\) are connected by an undirected edge if and only if they belong to different flows and there is a conflict between the two corresponding configurations.
\end{definition}

To insert a configuration \(c\) for flow \(f\) to \(\mathbf{G}(p)\), we add \(c\) to the vertex set of \(\mathbf{G}(p)\) and add an edge to all conflicting candidates of flows other than \(f\) in \(\mathbf{G}(p)\).
To remove a single configuration \(c\) from \(\mathbf{G}(p)\), we delete all edges between \(c\) and other configurations in \(\mathbf{G}(p)\), and remove \(c\) from the vertex set of \(\mathbf{G}(p)\).

\subsection{Global Traffic Plan}
From the perspective of the planning problem, a traffic plan \(p\) provides a phase and route for each flow \(f\) in ActiveF(\(p\)) such that all packet transmissions are isolated temporally and/or spatially, and thus satisfy the zero-queuing principle.
In other words: 
\begin{definition}[Traffic Plan]
	A traffic plan \(p\) is a set of conflict-free configurations for ActiveF(\(p\)) in \(\mathbf{G}(p)\) which contains exactly one configuration for each flow in ActiveF(\(p\)).
\end{definition}
We say that traffic plan \(p\) admits flow \(f\) if \(p\) contains a configuration for \(f\), i.e.,
\begin{equation}
	\operatorname{admits}(f,p)= \begin{cases}
		1: \text{if} \ \exists \ \text{configuration for \(f\) \(\in\) \(p\)} \\
		0: \text{if} \ \nexists \ \text{configuration for \(f\) \(\in\) \(p\)}
	\end{cases}.
\end{equation}

\begin{definition}[Independent Set]
	Denote by \(\mathcal{V}\) the vertices, and denote by \(\mathcal{E}\) the edges in \(\mathbf{G}(p)\), respectively.
	\(\mathcal{C} \subseteq \mathcal{V}\) is an independent subset of vertices in the conflict graph \(\mathbf{G}(p)\) iff \(\forall u,v \in \mathcal{C} : (u,v) \notin \mathcal{E} \).
\end{definition}

With the correspondence vertex \(\leftrightarrow\) configuration and edge \(\leftrightarrow\) conflict, it is easy to see that a traffic plan corresponds to a set of independent vertices in the conflict graph~\cite{falkTimeTriggeredTrafficPlanning2020}.

Note that \(\mathbf{G}(p)\) possibly contains many sets of different conflict-free configurations.
While it is intuitively clear that we prefer to find a set of conflict-free configurations that admits all flows, this may not always be possible.
Therefore, we specify the objective that guides the search for the configuration in Sec.~\ref{sec:problem_statement}.

\subsection{Transition Interference}
\begin{figure}
    \centering
    \includegraphics[
    trim={.6cm .6cm .6cm .6cm},clip,
    width=0.9\linewidth
    ]{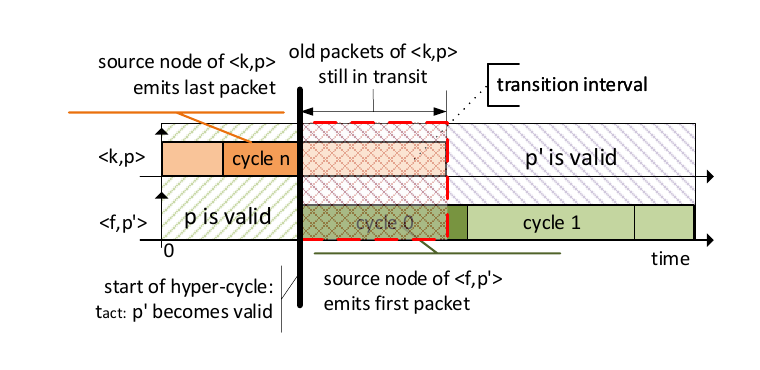}
    \caption{
        \markChanges{
            The possibility of transition interference between old packets of flow \(k\) from plan \(p\) and new packets of flow \(f\) from \(p'\) is limited to the \emph{transition interval}.
        }
    }
    \label{fig:reconfigurationconflict}
\end{figure}

Assume that we have a new traffic plan \(p'\) and want to replace the traffic plan \(p\) with \(p'\).
We call the time a new traffic plan \(p'\) becomes valid its activation time \(t_\text{act}\).
To ensure that no source node is in the middle of a packet transmission at the activation time of \(p'\), and since the flows in Active(\(p\)) may have different cycle times, we choose the start of a fresh hyper-cycle of all active flows in ActiveF(\(p\)) as activation time, i.e., \(t_\text{act} = t_0 + k \cdot \operatorname{lcm}_f(t_\text{cycle}(f))\).
From the perspective of an active flow it is perfectly fine for \(t_\text{act}\) to be very far---many transmission cycles---in the future, because all flows in ActiveF(\(p\)) remain active with their established configurations from \(p\) until \(t_\text{act}\).

Then, at \(t_\text{act}\) the new traffic plan becomes valid.
This means, from \(t_\text{act}\) on, we would like \emph{all} the packets of each flow \(f\) in ActiveF(\(p'\)) to be transmitted according to \(\operatorname{config}(f,p')=(f,\phi',\pi')\).
If all packets sent according to the old plan \(p\) are delivered to the destination nodes before \(t_\text{act}\), this would simplify our problem a lot.
It implies that the network would be empty at \(t_\text{act}\) which would allow us to compute each new traffic plan \(p'\) independently of \(p\).
In other words, under these circumstances computing a new traffic plan is equivalent to solving a new static traffic planning instance.

However, this simplified \enquote{solution} is inadequate because it severely restricts the network utilization.
Ultimately, for each packet to arrive within the transmission cycle it was sent by the source node, \(\phi\) would need to be less than \(t_\text{cycle}-t_\text{e2e}-t_\text{trans}\).
This means, the transmission frequency and network utilization are inherently coupled to the network size.

Therefore, we provide a solution for the general case where there can be remaining yet-to-be-delivered packets from \(\langle k,p\rangle\), \(k \in \operatorname{ActiveF}(p)\) in the network when the new traffic plan \(p'\) becomes valid.
In the general case, we have to account for the possibility of \emph{transition interference}, i.e., flow interference between the remaining old packets from \(\langle k,p\rangle\), \(k \in \operatorname{ActiveF}(p)\) and the first packets from \(\langle f,p'\rangle\), \(f \in \operatorname{ActiveF}(p')\), when computing the new traffic plan \(p'\).

\begin{definition}[Transition Interference]
	Let \(p\) denote the traffic plan valid before \(t_\text{act}\) and let \(p'\) denote the new traffic plan with activation time \(t_\text{act}\).
	Transition interference is defined to be between \(\langle k, p\rangle\), \(k \in ActiveF(p)\) and \(\langle f, p' \rangle\), \(f \in ActiveF(p')\) if any packet sent by the source node of \(\langle k, p\rangle\) prior to \(t_\text{act}\) is buffered due to a packet sent by the source node of \(\langle f, p' \rangle\) from \(t_\text{act}\) on, or vice versa.
\end{definition}
We can algorithmically check whether transition interference occurs between \(\langle k, p\rangle\) and \(\langle f, p' \rangle\):
If the route of \(\langle k, p\rangle\) and the route of \(\langle f, p'\rangle\) have common links, and any packet from \(\langle k, p\rangle\) sent before \(t_\text{act}\) is scheduled for transmission at any point in time when any packet sent by \(\langle f, p'\rangle\) from \(t_\text{act}\) on is scheduled on a common link, transition interference occurs.

Transition interference is restricted to a possible zero-length time interval, the \emph{transition interval} \(t_\text{transit}\), which starts at \(t_\text{act}\) and is upper-bounded by the end-to-end delay of \(\langle k, p\rangle\), cf. Fig~\ref{fig:reconfigurationconflict}.
Now we can also concretize the \enquote{purging} of obsolete flow information sets (cf. Sec.~\ref{sec:netnode_capabilities}):
It is easy to see that once the transition intervals of all flows in ActiveF(\(p\)) have expired, the controller can safely delete all unused flow information sets associated with \(p\) from the nodes in the network.

In the following, we only consider the case where transition intervals of different traffic plan updates are mutually exclusive, i.e., we consider traffic plan \(p\) and its immediate successor \(p'\).


\subsection{Problem Statement}
\label{sec:problem_statement}
Considering the many facets of dynamic traffic planning, our focus is how to actually compute the traffic plan and the traffic plan (updates) for an evolving set of real-time traffic flows.

Let \(p\) be the traffic plan for the flows in ActiveF(\(p\)) computed from \(\mathbf{G}(p)\), i.e., the packets of flows in ActiveF(\(p\)) are transmitted according to \(p\).
Now applications request that a set of flows ReqF(\(p\)) is added to the network.
Our goal is to compute a new traffic plan \(p'\).

To compute \(p'\), we have to construct the new conflict graph \(\mathbf{G}(p')\) and search for a set of conflict-free configurations (i.e., independent vertices) in \(\mathbf{G}(p')\).
In our case, we want to ensure that once the planner added a flow to ActiveF(\(p\)), this flow remains active until the application itself indicates the flow can be removed.
This means, the planner shall not unsolicitedly evict any flow from ActiveF(\(p\)).
Consequently, the new traffic plan shall include a configuration for \emph{every} flow in ActiveF(\(p\)) and as many new flows as possible from ReqF(\(p\)).
The different importance of active flows and new flows results in the following objective.
\begin{definition}[Traffic Planning Objective]
	Let \(\mathbf{G} (p')\) be a conflict graph which contains candidates for all flows in \(\operatorname{ActiveF}(p) \cup \operatorname{ReqF}(p)\) with \(\operatorname{ActiveF}(p) \cap \operatorname{ReqF}(p) = \emptyset\).
	We want to find a new traffic plan \(p' \subseteq \bigcup_{f} \operatorname{cand}(f,p')\) that maximizes the objective
	\begin{equation}
        \begin{split}
		\max_{p'} \sum_{f \in \operatorname{ActiveF}(p)} \operatorname{admits}(f,p')\\ +\sum_{f \in \operatorname{ReqF}(p)} \frac{\operatorname{admits}(f,p')}{\left|\operatorname{ActiveF}(p) \cup \operatorname{ReqF}(p)\right|}.
        \end{split}\label{eq:objective_weighted_max_color}
	\end{equation}
\end{definition}
The factor \({1}/{\left|\operatorname{ActiveF}(p) \cup \operatorname{ReqF}(p)\right|}\) in Eq.~\ref{eq:objective_weighted_max_color} discounts the relative importance of flows in \(\operatorname{ReqF}(p)\).
This means, any flow in \(\operatorname{ActiveF}(p)\) is more \enquote{valuable} than all flows \(\operatorname{ReqF}(p)\) combined.
From a graph-theoretic perspective, this is a specific colorful independent vertex set problem where we want to find a set of independent vertices and each color has either unit weight or a weight inversely proportional to the total amount of colors.

\section{Constructing the Conflict Graph}
\label{sec:constructingcg}
Next, we explain how the planner constructs the new conflict graph \(\mathbf{G}(p')\) when processing ReqF(\(p\)).
With offensive planning, the planner has to consider both flow versions \(\langle f,p \rangle\) and \(\langle f,p'\rangle\) for each active flow \(f\), where \(p'\) indicates the new plan replacing \(p\). 

\subsection{Overview}
Remember that after installing \(p'\), packets from both, \(\langle f, p\rangle\) and \(\langle f,p'\rangle\), may be in the network for some time.
Therefore, it is not sufficient to only avoid conflicts between any new flow versions \(\langle f, p' \rangle\) which would suffice in the static case.
Instead, the planner additionally has to ensure that any \(\langle k, p\rangle\) does not interfere with any other flow \(\langle f, p'\rangle\).
This leads to the following requirements:
\begin{description}
	\item[R1] For all \(k,f \in \operatorname{ActiveF}(p')\), \(\langle k,p'\rangle\) does not interfere with \(\langle f,p'\rangle\) for \(k \neq f\), i.e., active flows of the new plan do not interfere.
	\item[R2] For all \(k \in \operatorname{ActiveF}(p)\) and \(f \in \operatorname{ActiveF}(p')\setminus \operatorname{ActiveF}(p)\), \(\langle k,p\rangle\) does not interfere with \(\langle f, p'\rangle\).
	That is, the active flows of the old plan do not interfere with the flows added to the new plan. 
	\item[R3] For all \(k,f \in \operatorname{ActiveF}(p)\),  \(\langle k,p \rangle\) does not interfere with \(\langle f,p'\rangle\).
	That is, the active flows of the old plan do not interfere with their versions in the new plan.
	If \(f=k\), then we consider the old and new version of the same flow, \(\langle f,p\rangle\) and \(\langle f,p'\rangle\).
	Since packets from \(\langle k,p\rangle\) and \(\langle f,p'\rangle\) may populate the network at the same time, we have to ensure that no transition interference occurs.
\end{description}
R1 is a basic requirement, which is sufficient for defensive planning.
Allowing for reconfigurations of active flows, however, we must add R2 and R3 to ensure the zero-queuing constraint. 
In order to fulfill R1 the planner has to make sure that the configurations associated with flows in ActiveF(\(p'\)) are not in conflict.
To this end, the planner first has to \emph{add candidate configurations} for each flow in ReqF(\(p\)) to the conflict graph.
The expanded conflict graph is the basis for solving the planning problem as stated above.

Now let us see, how additionally R2 and R3 can be guaranteed.
Note that the interference of two flows/flow versions implies that their configurations are conflicting, while the opposite is not true.
For example, if we ensure that the packets of one flow or flow version only enter the network after all packets of another flow have already left the network, these two flows do not interfere even if their configurations are conflicting.
Therefore, we can fulfill R2 by \emph{delaying the activation} of the added flows until all packets that belong to flow versions of \(p\) have been delivered:
Assume the transmission of packets of each flow \(\langle k,p\rangle\) with \(k \in \operatorname{ActiveF}(p)\) is stopped at time \(t\).
All packets sent by these flows have left the network by \(t + \tau\), where \(\tau\geq \max_k (t_\text{e2e}(k) - t_\text{cycle}(k))\) is greater or equal to the largest transition interval \(t_\text{transit}\).
If the transmission of each \(\langle f,p'\rangle\) with \(f \in \operatorname{ActiveF}(p')\setminus \operatorname{ActiveF}(p)\) does not start before  \(t + \tau\), then no flow \(\langle k,p\rangle\) interferes with any flow added to \(p'\), even if their configurations conflict.
To put it differently,  we prevent transition interference between flows from \(k \in \operatorname{ActiveF}(p)\) and \(f \in \operatorname{ActiveF}(p')\setminus \operatorname{ActiveF}(p) \subseteq \operatorname{ReqF}(p)\) from ever happening by adjusting the activation time of new flows, i.e., outside of the actual computation of \(p'\).
Simplified speaking: the controller tells the source nodes of the new flows that their actual activation time is not \(t_\text{act}\), but \(t_\text{act} + \lceil \frac{t_\text{transit}}{t_\text{cycle}}\rceil \cdot t_\text{cycle}\) with \(t_\text{transit}\) the duration of the longest transition interval of any flow in ActiveF(\(p\)).
In the example in Fig.~\ref{fig:reconfigurationconflict}, the first \enquote{usable} transmission cycle for \(\langle f,p'\rangle\) is cycle 1.
For the source node of a new flow this has the same effect as if it took a few \SI{}{\micro\second} longer to compute \(p'\).
Therefore, fulfilling R2 requires no special treatment when constructing the conflict graph.

Note that this only applies for transition interference involving flows that are newly added to ActiveF(\(p'\)).
Fulfilling R3 the same way as R2 is definitely not desirable since delaying data packets of an active flow might degrade QoS substantially.
Figuratively speaking, our solution for R2---\enquote{turning on the machines a bit later}---does not work for R3 because \enquote{the machines are already running}.
Therefore, we have to make sure that the planner cannot select configurations for active flows which conflict with the old configurations of these flows.
We achieve this by \emph{locking those configurations} that the planner must not select for the new plan.
Note that the configurations to be locked depend on the old configuration of the active flows.
Since the configuration of an active flow may change from plan to plan, configurations are unlocked each time after the new plan has been computed.

\subsection{Adding Candidate Configurations}
\label{sec:constructingcg:adding_candidates}
Since \(\mathbf{G}(p)\) only contains candidates for flows in ActiveF(\(p\)), the planner has to generate and insert new candidate configurations for all flows in ReqF(\(p\)) to \(\mathbf{G}(p')\).
In principle, the planner could add all candidates for each flow to \(\mathbf{G}(p')\).
However, we know from the static flow planning problem~\cite{falkTimeTriggeredTrafficPlanning2020} that we can find a traffic plan for all flows even if the conflict graph contains just a subset of all possible candidates of each flow.
We exploit this and add only a limited number (upper-bounded by a parameter \(n_\text{ub}\)) of candidate configurations for each flow in ReqF(\(p\)) to the conflict graph.

Due to the nature of the traffic planning problem, it
is often inherently difficult to identify promising candidate configurations a-priori.
In this paper, we therefore use the following heuristic:

\markChanges{
For each flow, the planner has a stateful configuration generator that walks through the \(\phi\)-\(\pi\)-space (cf. Fig.~\ref{fig:configurationgeneration}) and returns a new configuration each time it is invoked.
We can, figuratively speaking, think of a flow's \(\phi\)-\(\pi\)-space as a two-dimensional plane, where a configuration is a discrete point at a coordinate formed by the phase and path identifier of the configuration.
Each invocation of the generator then corresponds to \enquote{moving} to the next point, which, in turn, generates the respective configuration to be added to the conflict-graph.
In this context, the current coordinate is our generator state, and the rule for determining the next coordinate defines the order in which the configurations of a particular flow are added to \(\mathbf{G}(p')\).
}

\markChanges{
In our case, starting with \(\phi=0\), \(\pi=0\), \(\pi\) is incremented until all configurations for the current value of \(\phi\) have been covered (i.e., moving \enquote{upwards} in the \(\pi\)-dimension), before increasing \(\phi\) by \(\Delta \phi\) (i.e., jumping in the \(\phi\)-dimension).
We use for \(\Delta \phi\) the 75-th percentile of the transmission duration of all flows in \(\mathbf{G}(p')\), hence we expect that one to two phase increments by the planner suffice to resolve a possible interference on a single link.
If \(\phi + \Delta \phi\) exceeds the allowed range for \(\phi\), \(\phi\) is reset to the next, lowest uncovered phase-value, cf. Fig~\ref{fig:configurationgeneration}.
}

\markChanges{
In other words, every time we need a new configuration for a particular flow, we invoke its configuration generator.
The configuration generator then returns the next configuration in a sequence of this flow's candidate configurations obtained by sweeping (possibly in multiple passes along the \(\phi\)-range) over the \(\phi\)-\(\pi\)-space.
}

\begin{figure}
	\centering
	\includegraphics[
	clip,
	]{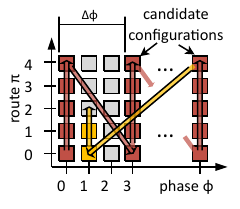}
	\caption{
        \markChanges{
            To obtain a new candidate configuration, the planner \enquote{moves} to a new point in the flow's discrete \(\phi\)-\(\pi\)-space and adds the visited candidate configuration corresponding to the current coordinates to the conflict graph.
            Our heuristic makes the planner walk as indicated by the arrows, i.e., the planner cycles through all candidate paths for the current phase \(\phi\), before jumping to the next phase.
        }
    }
	\label{fig:configurationgeneration}
\end{figure}

While it is obvious that we have to add candidates for flows in ReqF(\(p\)), we can also increase the set of candidates for flows in ActiveF(\(p\)) in \(\mathbf{G}(p')\).
The larger cand(\(f,p'\)), the more of the solution space is covered, and the more alternatives for an active flow \enquote{to make way} for flow in ReqF(\(p\)) exist.
Here, we add more configurations for each flow in ActiveF(\(p\)) when processing ReqF(\(p\)) until the configuration generator has traversed the \(\phi\)-range of a flow for the first time.
After that point, new configurations for active flows are only added if the planner could not admit all flows from the previous request,
i.e., a simple feedback loop controls the growth of the conflict graph.

In principle, we are free to choose how we want to traverse the \(\phi\)-\(\pi\)-space when adding configurations to the conflict graph.
This means, that the number of available configurations in the \(\phi\)-\(\pi\)-space is decoupled from the size of the conflict graph.
For example, if we want a time resolution of \SI{1}{\nano\second} instead of \SI{1}{\micro\second}, we can scale \(\Delta \phi\) accordingly, e.g., instead of advancing three steps in the \(\phi\)-direction, we advance 3000 steps and end up with a similar sized conflict graph.

This decoupling is an important feature because the effort to add a single configuration \(c\) to \(\mathbf{G}(p')\) scales with the number of configurations \(\mathcal{V}\) already in the conflict graph (\(\mathscr{O}(\mathcal{V})\)) since we have to check for each existing configuration already in \(\mathbf{G}(p')\) whether it conflicts with \(c\).
However, we do \emph{not} have to pay the \emph{total} cost of constructing a conflict graph with \(\mathcal{V}\) configurations, which scales with \(\mathscr{O}(\mathcal{V}^2)\).
When constructing \(\mathbf{G}(p')\), we start with \(\mathbf{G}(p)\).
Hence, the cost of adding a configuration effectively gets distributed over the lifetime of a flow.
The planner could also add the additional configurations for flows in ActiveF(\(p\)) when idle, i.e., while no ReqF(\(p\)) has to be processed, to further accelerate the flow-request response time.
Since the lifetime of active flows may differ, the planner should balance the number of candidates per flow in the long-term, e.g., by adding more configurations for younger flows.

\begin{figure}
	\centering
	\includegraphics[trim={.6cm .6cm .6cm .6cm},clip,
	width=0.9\linewidth]{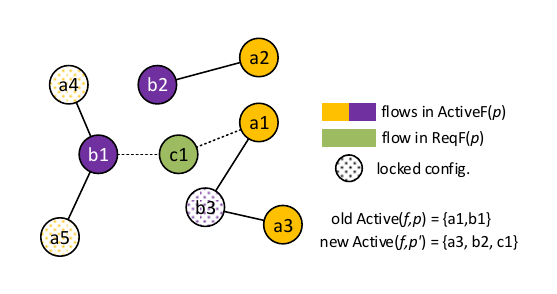}
	\caption{
        \markChanges{
            Example: Configuration \(c1\) for a new flow \(f_c\) conflicts with the active configurations \(a1\) and \(b1\) of \(f_a\) and \(f_b\), respectively.
            To admit \(f_c\), we need to reconfigure \(f_a\) and \(f_b\).
            Due to transition interference, all 1-hop neighbors of active configuration \(a1\) and \(b1\) are locked in \(\mathbf{G}(p')\), so only \(a2\), \(a3\), and \(b2\) remain as options during the reconfiguration.
        }
    }
	\label{fig:interflowconflicts}
\end{figure}

\subsection{Locking Configurations}
\label{sec:constructing:locking}
\(\mathbf{G}(p')\), generated as described in Sec.~\ref{sec:constructingcg:adding_candidates}, may include configurations that cause interference between the old and new versions of active flows, violating R3.
We yet have to make sure that the planner cannot select configurations for active flows that violate R3.
In detail, the following two conditions have to hold:
\begin{enumerate}
	\item For all \(f \in \operatorname{ActiveF}(p)\), there exists no configuration \(c \in \operatorname{cand}(f,p')\) that results in a new version \(\langle f,p'\rangle\) that interferes with \(\langle f, p\rangle\),  i.e., the planner cannot assign a new configuration to \(f\), such that the old and new version of \(f\) interfere.
	\item For all \(f,k \in \operatorname{ActiveF}(p)\), \(k \neq f\) there exists no configuration \(c \in \operatorname{cand}(f, p')\) that results in a new version \(\langle f, p'\rangle\) which interferes with \(\langle k,p\rangle\).
	That is, the planner cannot assign a new configuration to \(f\) which results in interference with an old version of any other active flow \(k\).
\end{enumerate}
We achieve both conditions by locking all those configurations in \(\mathbf{G}(p')\) which would violate one of the conditions if selected by the planner.
That is, informally speaking, we can think of a locked configuration in \(\mathbf{G}(p')\) as being \enquote{hidden} from the planner during the computation of \(p'\).

To fulfill the first condition, for each \(f \in \operatorname{ActiveF}(p)\) we have to visit each \(c \in \operatorname{cand}(f,p')\) in \(\mathbf{G}(p')\) to check whether a new version of \(f\) that uses the visited candidate \(c\) results in transition interference with the old version that uses \(\operatorname{config}(f,p)\).
If transition interference would occur, we have to lock the candidate.
Obviously, this ensures that the planner cannot select a candidate \(c\) for \(\operatorname{config}(f,p')\) which could cause interference with \(\langle f,p\rangle\).
Note that the configuration graph includes only edges for conflicting configurations of different flows.
Therefore, we maintain an additional data structure that provides efficient access to the candidates of a particular flow. 

To fulfill the second condition, for each \(k \in \operatorname{ActiveF}(p)\) we have to visit all 1-hop neighbors of config(\(k,p\)) in the conflict graph, i.e., all candidate configurations of any flow other than \(k\) that conflict with config(\(k,p\)).
If a visited neighbor \(c\) is a configuration of another flow \(f\) in ActiveF(\(p\)), the planner checks for transition interference between a new version of \(f\) with configuration \(c\) and \(\langle k, p \rangle\).
If there would be interference, this configuration \(c\) is locked.
Obviously, this prevents the planner from selecting for any other \(f \in \operatorname{ActiveF}(p)\) a candidate configuration \(c\) for config(\(f,p'\)) which would cause interference.
Note, here we just have to follow the edges of config(\(k,p\)) to access the candidates potentially to be locked, cf. Fig.~\ref{fig:interflowconflicts}.

\subsection{(Permanently) Pinning Flows}
To prevent the planner from ever reconfiguring an active flow \(f\), we can \emph{permanently} remove all configurations of \(f\) other than config(\(f,p\)) from \(\mathbf{G}(p)\)---\emph{pinning} \(f\) to its configuration for its whole lifetime.
Obviously, this ensures config(\(f,p\))=config(\(f,p'\)) for future updates.
In contrast, locking only temporarily excludes some candidates from \(p'\).

\section{Computing the New Traffic Plan}
\label{sec:computing_traffic_plan}
Once the new conflict graph has been constructed by adding additional configurations and locking currently \enquote{forbidden} candidate configurations, we have to compute the traffic plan.
From a graph-theoretic perspective, we can reduce the maximum (colorful) independent vertex set problem to the problem of finding a traffic plan which maximizes Eq.~\ref{eq:objective_weighted_max_color}.
This means, unsurprisingly, that computing a new traffic plan remains an NP-hard problem~\cite{garey_computers_1979-1}.
Therefore, we present a novel heuristic, namely, the \emph{Greedy Flow Heap Heuristic} (GFH).
The name draws from the fact that the objective in Eq.~\ref{eq:objective_weighted_max_color} improves with every additional flow included in the new traffic plan.
Next, we describe the GFH in detail, and explain our strategy to compute the new traffic plan.

\subsection{Greedy Flow Heap Heuristic}
From a birds-eye view, GFH is an iterative greedy approach.
\markChanges{
It splits ActiveF(\(p\)) and ReqF(\(p\)) into flow subsets and successively processes the different subsets in an order that is supposed to maximize the objective function (cf. Eq.~\ref{eq:objective_weighted_max_color}).
For every flow set, GFH iteratively selects the next flow, computes ratings for the flow's configurations in \(\mathbf{G}(p')\), and adds the configuration with the best rating to an intermediary set of conflict-free configurations \(\mathcal{C}\).
}

During the GFH execution, \(\mathcal{C}\) is always an independent vertex set in \(\mathbf{G}(p')\).
For the same flow \(f\) there may be multiple configurations in the final \(\mathcal{C}\) returned by the GFH, but we can only include one of these (we need only one route and phase) in a traffic plan.
In this case, we arbitrarily select one configuration for every flow with multiple configurations in \(\mathcal{C}\) since all configurations in \(\mathcal{C}\) are conflict-free.

Before explaining the greedy strategy and how configurations are rated, we introduce some terminology:

\begin{itemize}
	\item \(\mathcal{C}\) \emph{admits} flow \(f\), if \(\mathcal{C}\) contains at least one candidate configuration for \(f\).
	\item A configuration \(c \in \mathbf{G}(p')\) is \emph{shadowed} if at least one of its neighbors is included in \(\mathcal{C}\).
	Thus, adding a configuration to \(\mathcal{C}\) shadows all its neighbors.
	Shadowed configurations are excluded from being added to \(\mathcal{C}\) since they conflict, i.e., are connected by an edge, with at least one of the configurations in \(\mathcal{C}\).
	\item A configuration is \emph{solitary} if it has no neighbors (i.e., no edge to other configurations) in \(\mathbf{G}(p')\) and thus cannot conflict with any other configuration.
	\item A configuration \(c\) is \emph{eligible} for selection by the GFH algorithm if and only if \(c\) is neither shadowed, nor \(c \in \mathcal{C}\).
\end{itemize}

In each run, the GFH attempts to find a \(\mathcal{C}\) that admits all flows as follows:
First, \(\mathcal{C}\) is initialized by adding all solitary configurations.
Thus, \(\mathcal{C}\) already admits all flows with solitary configurations without reducing the solution space.
Then, configurations for the remaining flows are iteratively added to \(\mathcal{C}\).
We will discuss the configuration selection strategy in Sec.~\ref{sec:algorithm:selecting_colors} and the configuration rating in Sec.~\ref{sec:algorithm:selecting_vertices}, including a small GFH example.

If \(\mathcal{C}\) does not admit all flows after its first run, the GFH has a \emph{re-run} mechanism which attempts to find a \(\mathcal{C}\) that admits more flows by additional \emph{re-runs} with modified starting conditions.
The number of re-runs (default \(n_\text{re-runs}=3\)) can be parameterized.
We discuss the re-run mechanism in Sec.~\ref{sec:algorithm:rerun_mechanism}.
Finally, the complexity of GFH is addressed in Sec.~\ref{sec:algorithm:complexity}.

\subsubsection{Configuration Selection Strategy}
\label{sec:algorithm:selecting_colors}
GFH being a greedy algorithm, the quality of the solution, i.e., how many flows can ultimately be admitted, relies on the strategy we use to add configurations to \(\mathcal{C}\).
GFH uses a hierarchical, iterative strategy, first deciding on the next flow to process, and then selecting a configuration for that flow.
Remember that any active flow contributes more to the objective than all new flows taken together.
To prevent that any configuration of a new flow shadows any active flow's configuration such that an active flow may not be admitted to \(\mathcal{C}\), we thus first process all flows in ActiveF(\(p\)), before searching for a configuration for new flows from ReqF(\(p\)).

The pseudo-code to process a particular (sub-)set of flows, say \(\operatorname{searchF}(p') \subseteq \operatorname{ActiveF}(p) \cup \operatorname{ReqF}(p)\), is given in the method addConfigPerFlow in Alg.~\ref{alg:addConfigPerFlow}.
\begin{algorithm}
    \caption{addConfigPerFlow}
    \label{alg:addConfigPerFlow}
    \Input{flows set \(\operatorname{searchF}(p')\), conflict-free configurations \(\mathcal{C}\)}
    \(\forall f \in \operatorname{searchF}(p'):\) check \(\operatorname{admits}(f,\mathcal{C})\)\;
    \label{alg:addConfigPerFlow:admitCheck}    
    create heap of not admitted flows\;
    \label{alg:addConfigPerFlow:heapCreation}    
    \While{heap $\neq \emptyset$}{
        \(f_\text{min}\) $\gets$ pop \(f_\text{min}\) with least eligible configurations from heap\;
        \(c_\text{sel} \gets \) eligible \(c\) of flow \(f_\text{min}\) with smallest \(\operatorname{computeShadowRating}(c)\)\;
        add \(c_\text{sel}\) to \(\mathcal{C}\), update eligibility of neighbors of \(c_\text{sel}\)\;
        \label{alg:addConfigPerFlow:update}
        update heap (remove completely shadowed flows and reorder)\;
    }
    \Return updated $\mathcal{C}$\;
\end{algorithm}

We start by creating a min-heap that contains all flows from \(\operatorname{searchF}(p')\) not already admitted by \(\mathcal{C}\) (cf. Alg.~\ref{alg:addConfigPerFlow}, line~\ref{alg:addConfigPerFlow:admitCheck}).
The flows in the min-heap are sorted according to the number of remaining eligible configurations, i.e., flows with less eligible configurations are on top of the heap.
Ties are broken with the higher total degree of the configurations in cand(\(f,p'\)), because a high total degree suggests a high chance that all configurations of the flow will soon become ineligible.
The unique flow identifier is used as final tie-breaker to get a deterministic algorithm.

Next, the heap is iteratively processed:
The flow on top of the heap---the one with the least remaining eligible configurations---is removed from the heap, and we add the flow's best-rated (cf. Sec.~\ref{sec:algorithm:selecting_vertices}) configuration to $\mathcal{C}$.
When a configuration is added to \(\mathcal{C}\) the heap is adjusted.
Any completely shadowed flow is removed from the heap and the remaining flows are reordered according to their new eligible configurations count.

\subsubsection{Rating Configurations}
\label{sec:algorithm:selecting_vertices}
The shadowRating value of a configuration is intended to capture the \enquote{cost} of selecting that configuration in terms of the remaining solution space. 
For example, adding a configuration to \(\mathcal{C}\) that shadows huge portions of the conflict graph is \enquote{expensive}.
Conversely, the cost is lower, the fewer neighbors are shadowed.
If neighbors are shadowed, it is preferable to shadow configurations of those flows with lots of remaining eligible configurations.
This intuition is encoded in the pseudo-code given in Alg.~\ref{alg:shadowRating}.
\begin{algorithm}
    \caption{computeShadowRating}\label{alg:shadowRating}
    \Input{configuration $c$}
    shadowRating = 0 \;
    \(\mathcal{F}_\text{neig} \gets\) set of flows of all neighbors of \(c\)\;
    \ForEach{\(f_n \in \mathcal{F}_\text{neig}\)}{
        \(\text{shadowCount}\) \(\gets\) number of \emph{eligible} configurations of flow \(f_n\) in neighborhood of \(c\) \;
        \(\text{eligibleCount}\) \(\gets\) total number of eligible configurations in \(f_n\)\;
        \label{alg:shadowRating:eligibleCount}
        $\delta \gets$ \(\text{shadowCount}\) / \(\text{eligibleCount}\) \;
        \lIf{$\delta$ = 1}{
            shadowRating+=$\alpha$
        }
        \lElse{
            shadowRating+=$\delta$
        }
    }
    \Return shadowRating of \(c\)\;
\end{algorithm}

We compute for each flow \(f_\text{n}\) with a configuration in the 1-hop neighborhood of \(c\) the share \(\delta\) of this flow's remaining eligible configurations that would be shadowed by picking \(c\).
Configurations which are already shadowed are not taken into account.
The shadowRating ordinarily amounts to the sum over the \(\delta\) values of \(c\)'s neighborhood.
If adding $c$ to $\mathcal{C}$ shadows all remaining eligible configurations of a flow (\(\delta=1\)), a very large constant $\alpha$ (default \(\alpha=1000\)) is added instead to discourage the selection of a configuration \(c\) which shadows the remaining configuration(s) of another flow.
Note that the shadowRatings can be computed in parallel since we only have to read the neighbors of $c$.
Further, the \textit{eligibleCount} calculation (cf. Alg~\ref{alg:shadowRating}, line~\ref{alg:shadowRating:eligibleCount}) can be optimized via caching.

To give a brief example of how GFH works, we revisit the conflict graph given in Fig.~\ref{fig:interflowconflicts}.
As there are no solitary configurations, GFH starts with \(\mathcal{C}=\emptyset\).
The active flows \(f_a\) and \(f_b\) are handled by the first addConfigPerFlow call.
Hence, the heap consists of these two flows where \(f_a\) has three and \(f_b\) has two eligible configurations, respectively.
First, the best configuration of flow \(f_b\) is calculated since \(f_b\) has less eligible configurations.
Since \(b1\) shadows all configurations of \(f_c\), \(b1\) gets a very high shadowRating whereas \(b2\) gets the best rating with \(0.33\) and is added to \(\mathcal{C}\).
Second, flow \(f_a\) is processed.
Here, \(a3\) has the lowest shadowRating (0.0 since \(b3\) is locked) and is added to \(\mathcal{C}\).
Note that \(a4\) and \(a5\) are locked, \(a2\) is shadowed and \(a1\) would shadow all configurations of \(f_c\).
Finally, addConfigPerFlow is called a second time with the new flows \(\operatorname{searchF}(p')=\{f_c\}\).
Thereby, \(c1\), which is the only configuration of \(f_c\), is added to \(\mathcal{C}\).
\markChanges{
Since all flows are admitted, there is no need to use the re-run mechanism in this example.
}

\subsubsection{Re-run Mechanism}
\label{sec:algorithm:rerun_mechanism}
\markChanges{
The re-run mechanism provided by the GFH (cf. Alg.~\ref{alg:gfh}) can be used to improve the number of flows admitted by \(\mathcal{C}\), e.g., if not all flows are in \(\mathcal{C}\).
In principle, a single run in the GFH corresponds to executing the body of the loop in Alg.~\ref{alg:gfh} (starting at line~\ref{alg:gfh:loop}) once.
}
\begin{algorithm}
    \caption{GFH}\label{alg:gfh}
    \Input{active flow set \(\operatorname{ActiveF}(p)\), new flow set \(\operatorname{ReqF}(p)\), \(n_\text{re-runs}\)}
    \(\mathcal{C} \gets \emptyset\) \tcp*{init. $\mathcal{C}$ (first run)} \label{alg:gfh:initC}
    \(\text{cache} \gets \emptyset\)\;
    \For{\(i = 0; i \leq n_\text{re-runs}; i\)++}{
        \label{alg:gfh:loop}
        \(\operatorname{naActiveF} \gets \operatorname{ActiveF}(p) \setminus \mathcal{C}\)\; \(\operatorname{aActiveF} \gets \operatorname{ActiveF}(p) \bigcap \mathcal{C}\)\;
        \(\operatorname{naReqF} \gets \operatorname{ReqF}(p) \setminus \mathcal{C}\)\; \(\operatorname{aReqF} \gets \operatorname{ReqF}(p) \bigcap \mathcal{C}\)\;
		\tcp{\(\mathcal{C} \gets\) solitary configurations:} 
        \(\mathcal{C} \gets \{v \in \mathcal{V} : degree(v) == 0\}\)\;
        \label{alg:gfh:resetC}
        \(\mathcal{C} \gets \text{addConfigPerFlow}(\operatorname{naActiveF},\mathcal{C})\)\;
        \label{alg:gfh:activeMissing}
        \(\mathcal{C} \gets \text{addConfigPerFlow}(\operatorname{aActiveF},\mathcal{C})\)\;
        \label{alg:gfh:activeF}
        \(\mathcal{C} \gets \text{addConfigPerFlow}(\operatorname{naReqF},\mathcal{C})\)\;
        \label{alg:gfh:newMissing}
        \(\mathcal{C} \gets \text{addConfigPerFlow}(\operatorname{aReqF},\mathcal{C})\)\;
        \label{alg:gfh:newF}
        \lIf{\(\mathcal{C} \ \text{admits} \ (\operatorname{ActiveF}(p) \cup \operatorname{ReqF}(p)) \)}{\Return \(\mathcal{C}\)
        }\lElse{
		\(\text{store } \mathcal{C} \text{ in cache}\) 
}
    }
    \Return \(\argmax_{\mathcal{C} \in \text{cache}} \text{TrafficPlanningObjective}(\mathcal{C})\) \tcp*{see Eq.~\ref{eq:objective_weighted_max_color}}
\end{algorithm}

\markChanges{
In each run, we call addConfigPerFlow once for every flow subset, using it as parameter searchF(\(p'\)).
Remember that the order of the flow subsets for which we execute addConfigPerFlow implicitly assigns priorities to those subsets.
Due to the greedy strategy, the earlier we include a flow \(f\) from \(\mathbf{G}(p')\) in searchF(\(p'\)) the likelier it is that a configuration for \(f\) is added to \(\mathcal{C}\).
We take advantage of this in the re-run mechanism to prioritize those flows which were not admitted to \(\mathcal{C}\) in the previous run.
Note that we still have to account for the higher importance of active flows compared to new flows (cf. Eq.~\ref{eq:objective_weighted_max_color}).
To this end, we split both, ActiveF(\(p\)) and ReqF(\(p\)), into two subsets each: one subset contains all flows which were \textbf{a}dmitted by \(\mathcal{C}\) in the previous run (prefixed by \enquote{a}, i.e., aActiveF, aReqF in Alg.~\ref{alg:gfh}), and another subset contains those flows which could \textbf{n}ot be \textbf{a}dmitted previously (prefixed by \enquote{na}, i.e., naActiveF, naReqF in Alg.~\ref{alg:gfh}).
Before each run, we reset $\mathcal{C}$ (cf. Alg.~\ref{alg:gfh}, line~\ref{alg:gfh:resetC}).
Then, we first process the two subsets corresponding to ActiveF(\(p\)), before processing the two subsets corresponding to ReqF(\(p\)) (cf. Alg.~\ref{alg:gfh} line \ref{alg:gfh:activeMissing}-\ref{alg:gfh:newF}).
}
To be exact, we call addConfigPerFlow on the four flow subsets in the following order: 1) flows from ActiveF(\(p\)) previously not admitted by \(\mathcal{C}\), 2) flows from ActiveF(\(p\)) previously admitted by \(\mathcal{C}\), 3) flows from ReqF(\(p\)) previously not admitted by \(\mathcal{C}\), and 4) flows from ReqF(\(p\)) previously admitted by \(\mathcal{C}\).

\markChanges{	
The first run constitutes a special case where the flows admitted by \(\mathcal{C}\) are not the result of a previous run but contain only those flows with solitary configurations.
Solitary configurations do not shadow any eligible configuration and consequently cannot affect the results during the re-runs.
By convention, we do not count the first run as a re-run, i.e, the loop can be executed up to \(n_\text{re-runs} + 1\) times.
Re-runs are performed either until \(\mathcal{C}\) admits all flows, or we run out of re-runs (default \(n_\text{re-runs}=3\)).
This can improve the objective value, but it is not guaranteed.
}
If we use up all re-runs, the GFH algorithm returns the $\mathcal{C}$ with the highest objective value seen so far.

\subsubsection{Complexity}
\label{sec:algorithm:complexity}

The worst-case runtime complexity of a single GFH run, i.e., adding the solitary configurations to \(\mathcal{C}\) and performing a single run of addConfigPerFlow, can be stated as \(\mathscr{O}(\mathcal{E}+\mathcal{V}\mathcal{F})\).
Here, \(\mathcal{E}\) is the number of conflicts, \(\mathcal{V}\) is the number of configurations in the conflict graph, and \(\mathcal{F}\) represents the number of flows.
In the following, we derive this complexity.

Adding all solitary configurations is an \(\mathscr{O}(\mathcal{V})\) operation performed exactly once.
Building the filtered heap is an \(\mathscr{O}(2\mathcal{F}\cdot \mathcal{V})\) operation (cf. Alg.~\ref{alg:addConfigPerFlow}, line~\ref{alg:addConfigPerFlow:heapCreation}).
The heap updates later will be done in logarithmic time, resulting in \(\mathscr{O}(\mathcal{V} \cdot \log \mathcal{F})\).
We require the \(\mathcal{V}\) factor in the heap operations to count the eligible configurations for the flow comparisons every time.
The shadowRating is calculated for every configuration once.
Thus, we have \(\mathscr{O}(\mathcal{V}(d + \mathcal{F} + 1))\) computations, where \(d\) is the highest degree in the graph.
The \(\mathcal{F}\) parameter comes from accumulating the \textit{shadowCount}.
This assumes that we calculate the \textit{eligibleCount} (cf. Alg.~\ref{alg:shadowRating}, line~\ref{alg:shadowRating:eligibleCount}) once for each flow in the heap and use the (cached) \textit{eligibleCount} values for the computations of all configurations of this flow, which takes \(\mathscr{O}(\mathcal{F}\cdot \mathcal{V})\) time.
Finding \(c_{\text{sel}}\) can be done space-efficient for every flow in \(\mathscr{O}(\mathcal{F}\cdot \mathcal{V})\) time.
The eligibility update of the neighbors of \(c_{\text{sel}}\) (cf. Alg.~\ref{alg:addConfigPerFlow}, line~\ref{alg:addConfigPerFlow:update}) for every flow is an \(\mathscr{O}(d\cdot \mathcal{F})\) operation.
When combining these sub-routines, GFH ends up with a total runtime complexity of \(\mathscr{O}(\mathcal{V}(d + 5\mathcal{F} + \log \mathcal{F} + 2) + d\mathcal{F})\), which can be simplified to \(\mathscr{O}(\mathcal{E} + \mathcal{V}\mathcal{F})\).

\markChanges{
So far, we more or less considered a single run of the GFH.
The re-run mechanism adds another factor \(r'=n_\text{re-runs}+1\) resulting in a simplified overall complexity of \(\mathscr{O}(r'\mathcal{E}+r'\mathcal{V}\mathcal{F})\).
}

\subsection{Rejecting and Removing Flows}
\label{sec:rejects_removals}
If ReqF(\(p\)) contains flows that are not admitted by \(\mathcal{C}\), the planner rejects the corresponding flows.
All candidate configurations for a rejected flow are consequently purged from \(\mathbf{G}(p')\).
Similarly, applications could indicate to the planner that active flows shall be removed from the network.
In this case the planner also purges the corresponding candidate configurations from the conflict graph.

\subsection{Optimization: Progressive Strategy for Offensive Planning}
\label{sec:progressive_strategy}
The planner employs locking and GFH in a two-phase meta-strategy for offensive planning:

In the first, defensive phase, for each flow in ActiveF(\(p\)) \emph{every} configuration in  \(\operatorname{cand}(f,p') \setminus \operatorname{config}(f,p)\) is locked in \(\mathbf{G}(p')\), i.e., we conserve the configuration of all active flows.
Then \(\mathbf{G}(p')\) exposes only one configuration per active flow and all candidates for new flows to the GFH.
If we already find a traffic plan \(p'\) for \(\operatorname{ActiveF}(p) \cup \operatorname{ReqF}(p)\), we are done, and usually saved computation time since the GFH considered fewer configurations.

Only if we cannot admit all new flows, we release the conservative locks (configurations which cause transition interference, cf. Sec.~\ref{sec:constructing:locking} remain locked) and expose the \enquote{full} conflict graph to the GFH.
This second phase widens the search space for the GFH at the expense of a longer runtime, and active flows now may be reconfigured.
If the GFH rejects any active flows in the second phase, we revert to the result from the first phase, which is guaranteed to include all active flows in the new traffic plan \(p'\).

\section{Installing the Plan}
\label{sec:installing_plan}
After having computed the new plan, the controller propagates the sub-plans to the nodes in the network.
The sub-plan sent to a particular infrastructure node defines how the nodes are supposed to route incoming packets, when to reserve transmissions windows for these packets, and when the new plan becomes active, i.e., the traffic plan's activation time \(t_\text{act}\).
In other words, after having computed the new plan, it is installed in two steps:

Firstly, we send the sub-plans, flow-information sets, and the activation time for the new plan via \emph{control channels} to all the nodes in the network while the old traffic plan is still active.
Secondly, at the respective \(t_\text{act}\), the nodes switch over to the new flow information sets.
Such a time-based two-step update pattern is, e.g., specified in~\cite{ieeecomputersocietyIEEEStandardLocal2018} to update the gate control lists of the Time-Aware Shapers.

A transition to the new traffic plan shall only happen if all nodes have agreed that they have received the new sub-plans \emph{and} will put their sub-plan into action at \(t_\text{act}\).
This means, each traffic plan update requires to solve the well-known consensus problem with a deadline that equals \(t_\text{act}\) (minus the time nodes need to process update-protocol packets in the worst-case).
However, it is well-known that it is impossible to solve the consensus problem in asynchronous systems~\cite{fischerImpossibilityDistributedConsensus1985}.
In particular, termination within bounded time cannot be guaranteed in such systems.
From a practical perspective, though, our specific network update problem has some properties which make this less of an issue:

Firstly, as soon as nodes or links fail, the consensus problem more or less becomes irrelevant since our network may not even work for the active flows anymore.
To handle this case, we need additional counter-measures such as redundancy or fail-over on the networking level in the first place, and these have to be incorporated into the traffic plans as well, which is out of the scope of this paper.

Secondly, we can implement \emph{static} real-time control channels between the controller and each network node.
These channels could be established when the network is initialized.
For example, we can even use the traffic planning approach described in this paper to set up these channels---\markChanges{before processing any time-triggered traffic flow requests issued by \enquote{regular} applications}---as pinned time-triggered flows such that their associated QoS is not subject to degradation.
With real-time channels and bounded update-protocol processing times in the nodes, we effectively execute the network update in a synchronous system.
Therefore, agreeing on the update within a deadline is possible.
But even \enquote{cheaper} real-time control channel implementations in terms of occupied network resources, e.g., non-time-triggered flows with bandwidth guarantees, are \enquote{safe} since the controller can set the activation time to a hyper-cycle boundary \enquote{far-enough} in the future.
Obviously, this involves a compromise between network resources spent on control channels and the time it takes to install an update.

As discussed in Sec.~\ref{sec:constructingcg}, the controller has to postpone the activation time for \emph{sources} of new flows after the transition interval whereas this is not necessary (and it would also be impractical) for infrastructure nodes that forward packets of new flows.
This means, there is one activation time for the source nodes of flows from ActiveF(\(p\)) and the infrastructure nodes.
The activation time for the source nodes of new flows however is delayed by \(\lceil \frac{t_\text{transit}}{t_\text{cycle}}\rceil \cdot t_\text{cycle}\) with \(t_\text{transit}\) the duration of the longest transition interval. 
While the flows from ActiveF(\(p\)) already use the new traffic plan, which also has been deployed to all infrastructure nodes by that time, the source nodes of new flows have to wait a little longer before they are allowed to start sending packets. 

\section{QoS Considerations}
\label{sec:qos}
The reconfiguration of active flows can degrade the QoS by introducing jitter.
Since we aim for deterministic real-time communication, we must also quantify and possibly contain the QoS degradation caused by reconfigurations to a level acceptable by the applications.
Next, we study how the reconfiguration of an active flow can degrade QoS.
The level of degradation depends on the \enquote{distance} of the flow's old and new configuration, i.e., the phase shift and the difference in the lengths of the routes.
After showing how these properties affect QoS degradation, we present a way how applications can control the degree of degradation.

\subsection{Computing QoS Degradation}
Without reconfiguration, the destination node of an active flow receives the next packet every \(t_\text{cycle}\) seconds after the reception of the previous packet, and packets are received in the order they were sent from the source node.
Now assume, a new traffic plan \(p'\) supersedes \(p\), and an active flow \(f \in \operatorname{ActiveF(p)}\) is reconfigured, i.e., config(\(f,p\))\(\neq\)config(\(f,p'\)).
W.l.o.g., the activation time of the new traffic plan \(p'\) is the start of the \(n+1\)-th cycle of \(f\).
\begin{figure}
	\centering
	\includegraphics[
	trim={.4cm .4cm .4cm .4cm},
	clip,
	width=\linewidth
	]{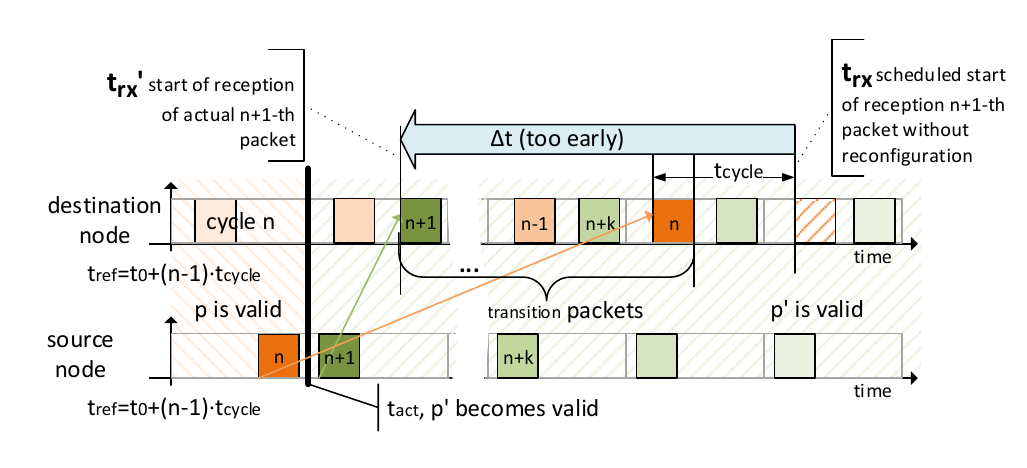}
	\caption{
        \markChanges{
            Reconfiguration of an active flow can temporarily cause jitter and packet reordering.
            Here, the packets with numbers \(n+1\) to \(n+k\) arrive \enquote{too early}, i.e., before packet \(n\).
        }
    }
	\label{fig:qos_degradation}
\end{figure}
Consequently, the point in time \(t_\text{rx}'\) when the destination node receives the \(n+1\)-th packet, which is sent with config(\(f,p'\)), may differ from the point in time \(t_\text{rx}\) when the destination would have received the \(n+1\)-th packet without reconfiguration, i.e., if the \(n+1\)-th packet had been sent with the old config(\(f,p\)), cf. Fig~\ref{fig:qos_degradation}.
We call the magnitude of this deviation \(\Delta t = t_\text{rx}' - t_\text{rx}\) the (reconfiguration) \emph{jitter}.
\(\Delta t\) can be computed with the following theorem.
\begin{theorem}[\(\Delta t\)]
	After a reconfiguration of flow \(f\) from \((f,\phi,\pi) = \operatorname{config}(f,p)\) to \((f,\phi',\pi') = \operatorname{config}(f,p')\), the arrival of the new packets at the destination node deviates by
	\(\Delta t = \left(\phi' - \phi\right) + \left(l(\pi') - l(\pi)\right) \cdot t_\text{perhop}\)
	from their respective scheduled arrival according to previous config(\(f,p\)) with \(l(\ast)\) denoting the number of infrastructure nodes on the candidate path with index \(\ast\).
\end{theorem}
This means, \(\Delta t\) is composed of a term \(\left(\phi' - \phi\right)\) expressing the phase-difference, and a term \(\left(l(\pi') - l(\pi) \right) \cdot t_\text{perhop}\) that expresses the difference between the traversal times of a packet on the old route and the new route.
\begin{proof}
	W.l.o.g. we use the start \(n\)-th cycle as reference time \(t_\text{ref}\).
	In the \(n\)-th cycle, the source node starts transmitting the \(n\)-th packet at \(t_\text{ref} + \phi\).
	The reception of this packet by the destination node starts at \(t_\text{ref} + \phi + t_\text{src} + t_\text{trans} + t_\text{prop} + l(\pi) \cdot t_\text{perhop} + t_\text{dst}\) (cf. Sec.~\ref{sec:zero-queuing}).
	If \(p\) remained valid, the reception of the \(n+1\)-th packet then would start \(t_\text{cycle}\) later at \(t_\text{rx}= t_\text{ref} + \phi + t_\text{src} + t_\text{trans} + t_\text{prop} + l(\pi) \cdot t_\text{perhop} + t_\text{dst} + t_\text{cycle}\).
	
	However, the new plan \(p'\) and thus config(\(f,p'\)) becomes valid at \(t_\text{ref} + t_\text{cycle}\), which is the start of the \(n+1\)-th cycle.
	The reception of the \(n+1\)-th packet sent by the source node with config(\(f,p'\)) starts at \(t_\text{rx}' = t_\text{ref} + t_\text{cycle} + \phi' + t_\text{src} + t_\text{trans} + t_\text{prop} + l(\pi') \cdot t_\text{perhop} + t_\text{dst}\) at the destination node.
	By inserting these values, we get 
	\(
	\Delta t = t_\text{rx}' - t_\text{rx} = \left(\phi' - \phi\right) + (l(\pi') - l(\pi) ) \cdot t_\text{perhop}
	\).
\end{proof}
If \(\Delta t > 0\), there is an additional delay between the last packet of \(\langle f,p \rangle\) and the first packet of \(\langle f,p' \rangle\).
\markChanges{
Intuitively, a reconfiguration that results in either a longer path or a phase increment (i.e., \(\phi'>\phi\)), or both, results in \(\Delta t > 0\).
}
If \(\Delta t < 0\), the new packets from \(\langle f,p' \rangle\) arrive too early.
\markChanges{
This happens, if either the reconfiguration results in phase decrement (\(\phi'<\phi\)), or a shorter route \(l(\pi')<l(\pi)\), or both of these.
}
Potentially, this can lead to a situation, where packets from \(\langle f,p' \rangle\) overtake the last packets from \(\langle f,p \rangle\) in the network if the relative end-to-end deadline is allowed to be greater than the cycle time.

We can upper bound the number of packets which possibly arrive in a different order than the order in which they were sent, and do not have an inter-arrival time \(t_\text{cycle}\) by \(n_\text{transition} \leq 2\cdot\left\lceil\frac{\left|\Delta t\right|}{t_\text{cycle}}\right\rceil\).
The term \(2\cdot\left\lceil\frac{\left|\Delta t\right|}{t_\text{cycle}}\right\rceil\) approximates the worst-case where old and new packets arrive interleavedly at the destination node until all old packets are delivered.

It is important to note that any individual packet that is sent by a source node will always arrive within the end-to-end deadline bounds.
Reconfigurations only affect the relative inter-arrival times in a packet train.

\begin{table*}[h]
    \caption{Overview of evaluation scenario parameters.}
    \label{tab:eval_properties}
    \centering
    \resizebox{0.7\linewidth}{!}{
        \begin{tabular}{l@{\hskip 0.2em} c@{\hskip 0.2em} c@{\hskip 0.2em} c@{\hskip 0.2em} c@{\hskip 0.2em} c@{\hskip 0.2em} c@{\hskip 0.2em} c@{\hskip 0.2em} c@{\hskip 0.2em}}
            \toprule
            Fig.& \textbf{\(\left|\operatorname{ActiveF}(p)\right|\)} & \textbf{\(\left|\operatorname{ReqF}(p)\right|\)} & flows to remove & \(n_\text{ub}\) & \(t_\text{cycle}\) [\SI{}{\micro\second}] &  \(n_\text{path}\) & network & \shortstack[c]{\(\sum\) ingress rate } \\
            \midrule
            \ref{fig:eval_const:runtime500} & 500 & 25 (poiss.) & 25 (poiss.) & \(\leq\) \{50,100\} & \{250,500,1000,2000\} & 3 & ring(64,3) & \(\sim\)615MB/s \\
            \ref{fig:eval_const:runtime800},\ref{fig:eval_const:rejects800} & 800 & 50 (poiss.) & 50 (poiss.) &  \(\leq\) \{50,100\} & \{250,500,1000,2000\} & 3 & ring(64,3) & \(\sim\)984MB/s\\
            \midrule
            \ref{fig:eval_compare:cummulative}, \ref{fig:eval_compare} & 250 & 25 (det.) & 25 (det.) & \(\leq\) 100 & \{200,250,500\} & 3 & ring(64,3) & \(\sim\)601MB/s\\
            \midrule
            \ref{fig:eval_topo:gfh}, \ref{fig:eval_topo:size49} & 500 & 50 (det.) & 50 (det.) & \(\leq\) 100 & \{250,500\} & 5 & var., 49 nodes & \(\sim\)984MB/s \\
            \ref{fig:eval_topo:gfh}, \ref{fig:eval_topo:size81} & 500 & 50 (det.) & 50 (det.) & \(\leq\) 100 & \{250,500\} & 5 & var., 81 nodes & \(\sim\)984MB/s\\
            \midrule
            \bottomrule
        \end{tabular}
    }
    \resizebox{0.7\linewidth}{!}{
        \begin{tabular}{l@{\hskip 0.5em} l}
            \toprule
            \(t_\text{trans}\) (packet size)  & \(\in\) \{\SI{1}{\micro\second},\SI{3}{\micro\second},\SI{5}{\micro\second},\SI{12}{\micro\second}\} (corresponds to \SI{125}{\byte} -- \SI{1500}{\byte} packets on \SI{1}{\giga\bit\per\second} links) \\
            flow clusters & \(\in \{1,2,4,8,16,32\}\) (subject to \(\left|\operatorname{ReqF}(p)\right|\) and num. flows to remove)\\
            \midrule
            network parameters &  processing delay \(t_\text{proc}=\SI{2}{\micro\second}\), propagation delay \(t_\text{prop}=\SI{1}{\micro\second}\)\\
            \bottomrule
        \end{tabular}
    }
\end{table*}

\subsection{Restricting QoS Degradation}
Applications can provide the planner with QoS degradation constraints in the form of thresholds on \(\left|\Delta t\right|\) and the number of affected packets.
If the destination node has no facilities to handle packet reorderings, e.g., has no buffer, then such constraints can be used to ensure that all packets are received in the sending order.
The QoS degradation constraint results in the following condition for \(\mathbf{G}(p')\):
\begin{description}
	\item[QoS]
    For each \(f \in \operatorname{ActiveF}(p)\), there exists no configuration \(c \in \operatorname{cand}(f,p')\), where the reconfiguration from \(\operatorname{config}(f,p)\) to \(c\) exceeds the threshold on \(\left|\Delta t\right|\).
\end{description}
Analog to Condition 1 from Sec.~\ref{sec:constructing:locking}, we can prevent reconfigurations that violate the QoS condition by \emph{locking}.
This can be achieved with minimal overhead since the planner anyway traverses \(c \in \operatorname{cand}(f,p')\), and only needs to check the QoS condition for each candidate \(c\) of \(f\) that is not locked already due to transition interference.

\section{Evaluation}
\label{sec:eval}
We evaluated a prototypical C++-implementation of the planner.
The evaluation scenarios consist of a sequence of planning rounds.
In each round, the planner processes a set of flow requests which can include requests for new flows as well as requests for the removal of active flows.
Besides the conflict graph adjustments, and the computation of the new traffic plan \(p'\), the planner also validates the absence of configuration conflicts in the new traffic plan in each round.

\subsection{Setup and Parameters}
By default, we used a desktop-grade computer with Intel i7-10700K (8 cores) and \SI{16}{\giga\byte} RAM for the evaluations.

Table~\ref{tab:eval_properties} gives an overview over the evaluation parameters which have been derived from typical industrial use cases~\cite{noauthor_iecieee_nodate,industrial_internet_consortium_time_2019}.
For each flow, we draw the values for \(t_\text{trans}\) and \(t_\text{cycle}\) uniformly at random from the respective set.
The number of new flows to add, i.e., \(\left|\operatorname{ReqF}(p)\right|\) and the number of flows to remove is either drawn from a truncated Poisson distribution (poiss.) with averages as stated in Tab.~\ref{tab:eval_properties}, or is deterministic (det.) for each planning round.
\markChanges{
The expected ingress traffic data rate, i.e., the aggregate of the traffic entering the network, is computed from the random distributions for \(t_\text{cycle}\) and \(t_\text{trans}\).
}
We place/remove the flows in clusters in the network to simulate control units connected to multiple sensors and actuators, where all flows in a cluster start or target one common \enquote{cluster} node.
The cluster sizes correspond to commonly found configurations of I/O bays of industrial control units.
For example, to add 25 new flows the request may contain three clusters of size 1, 8 and 16 flows, respectively, or any other combination that adds up to 25 flows.
By default, 20\% of the new flows are pinned permanently, and we limit \(\Delta t\) for each remaining active flow to \(t_\text{cycle}-t_\text{trans}\) such that packets arrive in order at the destination node.
\markChanges{
We omit specifying an explicit value for \(t_\text{e2e}\).
Instead, we use a k-shortest path algorithm which provides us with the candidate paths with the lowest end-to-end delays that are possible within the given topology.
}

In Sec.~\ref{sec:eval_const} and \ref{sec:eval_schedulability}, we consider circular networks with \(n\) nodes where each node is connected to the next \(k\) nodes in both directions (denoted by ring(\(n\),\(k\))).
\markChanges{
Besides being a common building block in, e.g., industrial automation networks~\cite{hellmannsScalingTSNScheduling2020a} or sensor networks, the regularity of ring(\(n\),\(k\)) topologies eliminates several sources of uncertainty that could, e.g., result in infeasible problem instances, or otherwise limit the problem space that we can evaluate.
For example, think of topologies with a bottleneck link connecting two sub-networks that---depending on how many flows have to cross that bottleneck---will result in massively different numbers of \enquote{acceptable} flows for the same topology.
}
In contrast, in ring(\(n,k\)) topologies, we can expect similar behavior for the flow placements since each node in the network has the same degree and there is an equal number of alternative paths between any pair of nodes for this topology.
In Sec.~\ref{sec:eval:topology}, we investigate the effects of different network topologies considering additional graph models, too.

In Tab.~\ref{tab:eval_properties}, column \(n_\text{ub}\) denotes the upper bound on how many candidates per flow the planner may add at most to \(\mathbf{G}(p')\).
Likewise, \(n_\text{path}\) is the upper bound on the number of candidate paths per flow.
Both, \(n_\text{ub}\) and \(n_\text{path}\), can be considered parameters of our planner which result in different trade-offs wrt. runtime and schedulability.
For example, for meshed networks, it is intuitive to expect better results with more candidate paths at the price of longer runtimes~\cite{jonathan_falk_optimal_2019}, but we expect a diminishing return with regard to the number of candidate paths due to shared sub-paths on the k-shortest paths.
Therefore, we use 3-5 candidate paths per flow, which already yields a very high schedulability in our evaluations (similar to \cite{falkTimeTriggeredTrafficPlanning2020}).

\begin{figure*}[!htpb]
    \centering
    \begin{subfigure}[t]{0.31\linewidth}
        \centering
        \includegraphics[
        trim={.0cm .0cm .0cm .0cm},
        height=1.4in
        ]{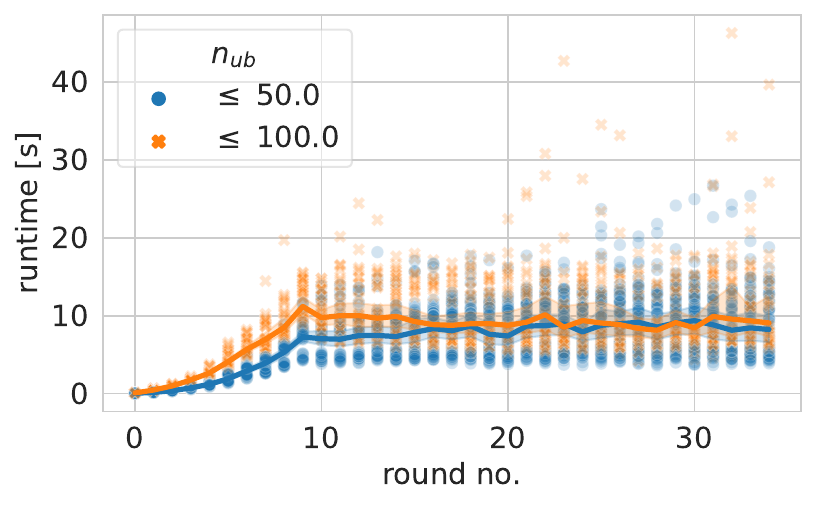}
        \caption{Runtime per round (\(\approx\)500 active flows).}
        \label{fig:eval_const:runtime500}
    \end{subfigure}
    \hspace{.5em}
    \begin{subfigure}[t]{0.31\linewidth}
        \centering
        \includegraphics[
        trim={.0cm .0cm .0cm .0cm},
        height=1.4in
        ]{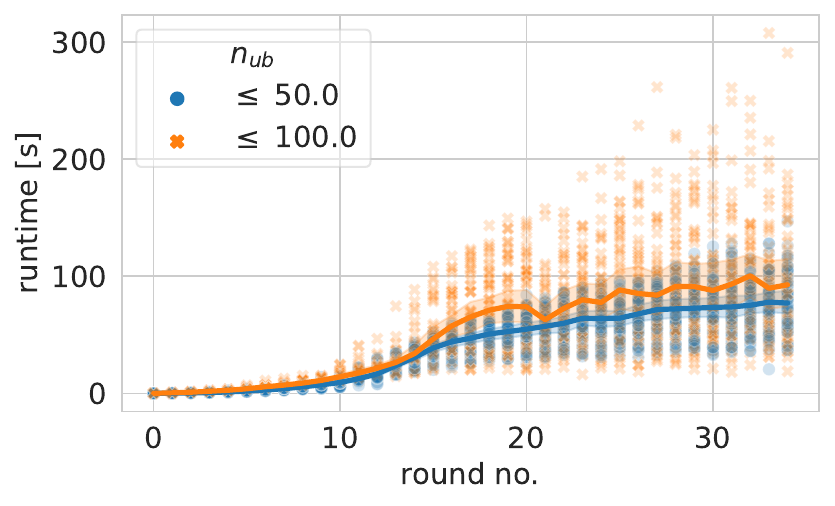}
        \caption{Runtime per round (\(\approx\)800 active flows).}
        \label{fig:eval_const:runtime800}	
    \end{subfigure}
    \hspace{.5em}
    \begin{subfigure}[t]{0.31\linewidth}
        \centering
        \includegraphics[
        trim={.0cm .0cm .0cm .0cm},
        height=1.4in
        ]{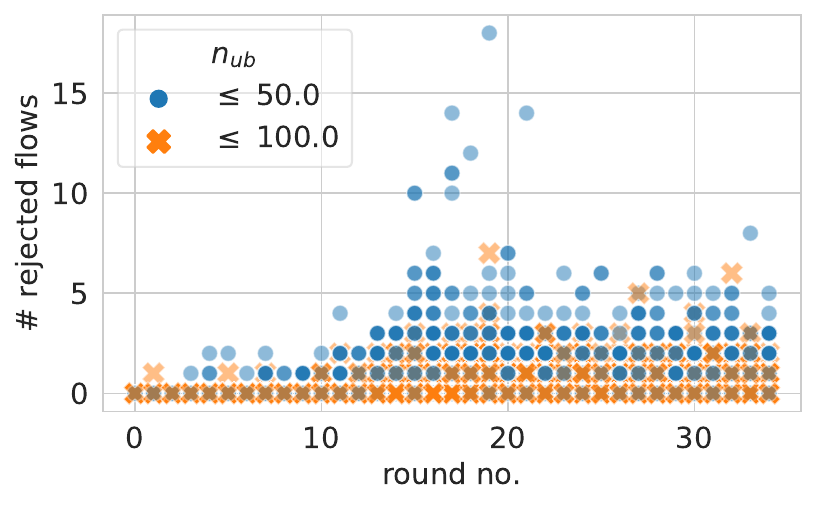}
        \caption{Rejections per round (\(\approx\)800 act. flows).}
        \label{fig:eval_const:rejects800}
    \end{subfigure}
    \caption{Runtimes and number of rejected flows per planning round (500 and 800 active flows).}
    \label{fig:eval_const}
\end{figure*}

\subsection{Runtime}
\label{sec:eval_const}
We evaluated 60 scenarios wrt. runtime for each combination of average active flows (500, 800) and per-flow candidate set increment limit (50, 100).
Fig.~\ref{fig:eval_const} shows the results for 35 planning rounds.

We plot the \emph{total} runtime for each planning round in Fig.~\ref{fig:eval_const:runtime500} and \ref{fig:eval_const:runtime800}, which includes the time to compute the candidate paths for new flows, all the operations on the conflict graph (adding, locking, pinning and removing candidate configurations) and the GFH runtime.
For these scenarios, we measured up to \SI{9.072}{\giga\byte} RAM usage by the planner.

In the first 10 (16) rounds, we always issued requests to add 50 flows (no flow-removals) to initialize the conflict graph from scratch until ActiveF(\(p\)) contains \(\approx\)500 (800) flows.
During these initialization rounds, we observe an increasing total runtime.
After the initialization rounds, the planner receives requests to remove and add flows.

In the smaller scenarios with \(\approx\)500 flows (cf. Fig.~\ref{fig:eval_const:runtime500}), the runtime per round plateaus (avg: \SI{8.24}{\second}\(\pm\)\SI{3.58}{\second}, \(n_\text{ub}\)=50; \SI{9.2}{\second}\(\pm\)\SI{4.03}{\second}, \(n_\text{ub}\)=100) after the initialization rounds.
In other words, here it takes on average less than \SI{10}{\second} to exchanges 25 active flows (and their configurations) against 25 new flows (and the respective new configurations) even though the planner can reconfigure \(\approx\)400 active flows.
Here, we also have a quite \enquote{stable} conflict graph size in the post-initialization rounds which results in the comparatively constant runtimes per round.

In Fig.~\ref{fig:eval_const:runtime800}, the runtimes for the bigger scenarios with \(\approx\)800 active flows are depicted.
Here, the runtime continues to grow after the 16 initialization rounds (max: \SI{146.83}{\second} for \(n_\text{ub}\)=50; \SI{307.72}{\second} for \(n_\text{ub}\)=100).
This behavior can be explained with Fig.~\ref{fig:eval_const:rejects800} where we plot for each scenario how many new flows the planner rejected in each round.
As discussed in Sec.~\ref{sec:constructingcg:adding_candidates}, after the first pass over the \(\phi\)-\(\pi\)-space the planner adds more candidates for \emph{active flows} if flows from the previous request had to be rejected.
This means, we actually observe the desired expansion of the solution-space since we have to fit \(\approx\)300 more flows in the same network as for the smaller scenarios,
This causes the planner to grow the conflict graph more aggressively---resulting in longer runtimes per round.

The effects of varying \(n_\text{ub}\) can also be observed: if the planner may add up to 100 additional candidates per flow, it can admit more new flows, but generally takes longer:
In the scenarios with 800 flows the planner rejected on average 0.84 (0.23) flows per request set after the initialization rounds for \(n_\text{ub}=50\) (\(n_\text{ub}=100\)) per round.

\begin{figure}
	\centering
	\includegraphics[
	trim={.0cm .0cm .0cm .0cm},
	height=1.4in
	]{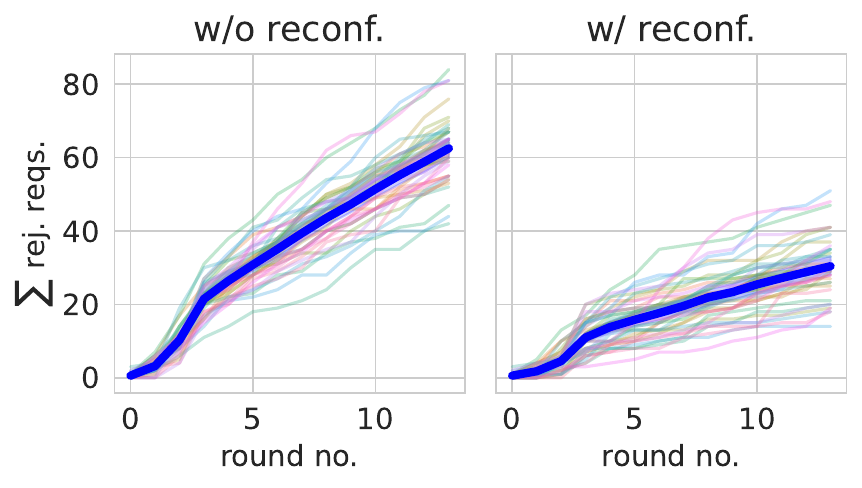}
	\caption{Defensive vs. offensive planning: Cumulative rejected flows (thin lines: per scenario, bold line: average).}
	\label{fig:eval_compare:cummulative}
\end{figure} 

\begin{figure*}[!htpb]
    \centering
    \begin{subfigure}[t]{0.32\linewidth}
        \centering
        \includegraphics[
        trim={.0cm .0cm .0cm .0cm},
        height=1.4in
        ]{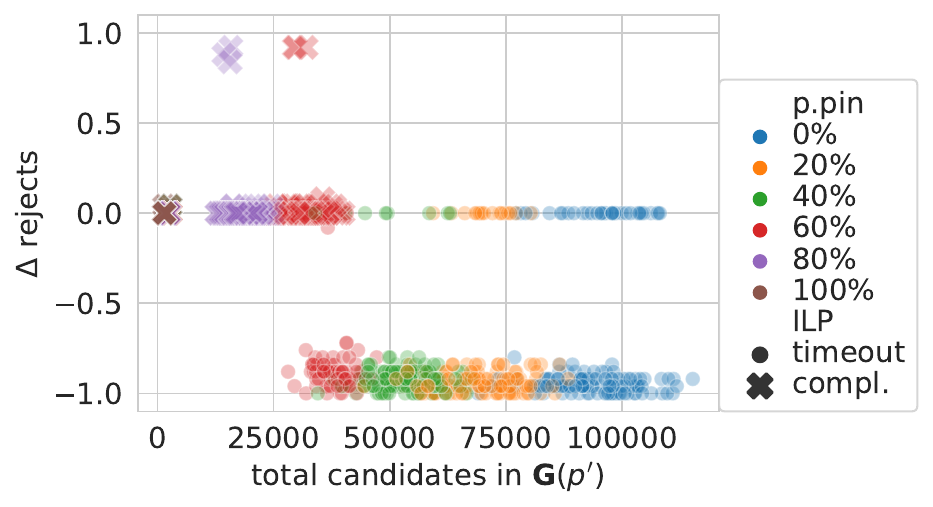}
        \caption{GFH vs. ILP: Relative difference wrt. rejected new flows for varying ratio of permanently-pinned active flows and \SI{5}{\min} runtime-limit.}
        \label{fig:eval_compare:delta}	
    \end{subfigure}
    \hspace{1em}
    \begin{subfigure}[t]{0.23\linewidth}
        \centering
        \includegraphics[
        trim={.0cm .0cm .0cm .0cm},
        height=1.4in
        ]{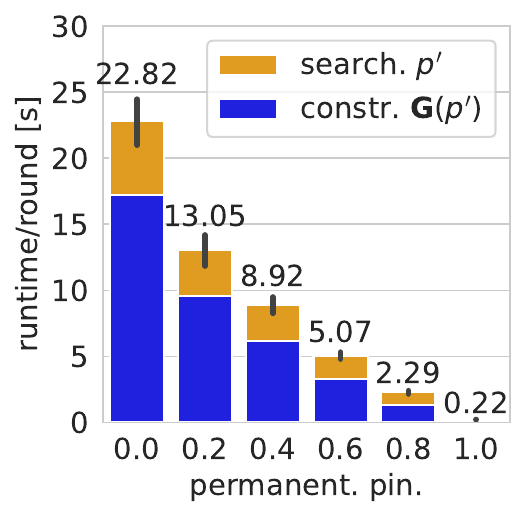}
        \caption{Runtimes per round for scenarios from Fig.~\ref{fig:eval_compare:delta} with GFH.}
        \label{fig:eval_compare:runtimes}	
    \end{subfigure}
	\hspace{1em}
	\begin{subfigure}[t]{0.32\linewidth}
		\centering
		\includegraphics[trim={.0cm .0cm .0cm .0cm},
		height=1.4in
		]{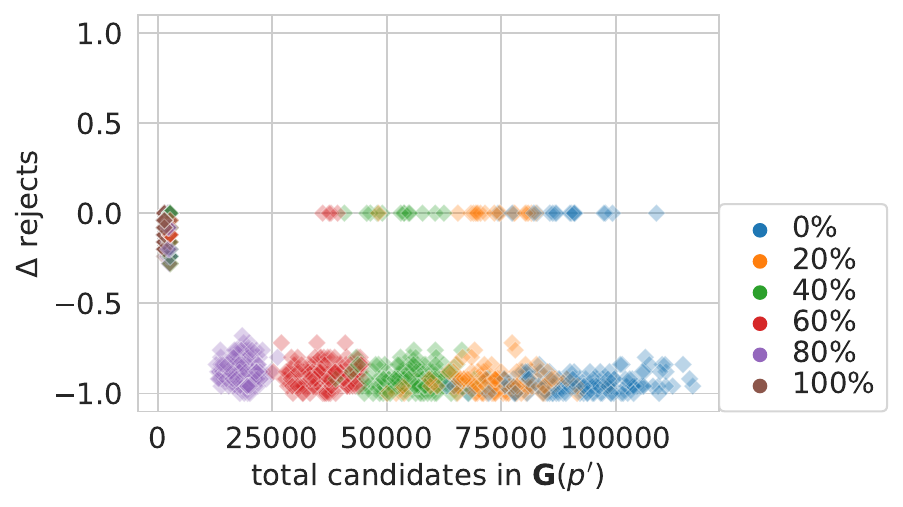}
		\caption{
			\markChanges{
				GFH vs. aLA: Relative difference wrt. rejected new flows for varying ratio of permanently-pinned active flows.
			}
		}
		\label{fig:eval_compare:delta_ala}
	\end{subfigure}
    \caption{Comparison of schedulability and runtime for varying ratios of non-reconfigurable (permanently-pinned) active flows.}
    \label{fig:eval_compare}
\end{figure*}

\subsection{Comparison}
\label{sec:eval_schedulability}
Next, we compare defensive and offensive traffic planning wrt. schedulability, and show that the new algorithms such as the GFH are required to make offensive traffic planning feasible.

The evaluation scenarios vary with respect to how many active flows are on average permanently pinned to their configuration.
We increase the ratio of pinned flows from 0\% to 100\% in 20\%-increments and evaluated 40 scenarios for each round and pinning-setting.
Here, there is no QoS restriction on \(\Delta t\) for each active flow that can be reconfigured.

\subsubsection{Defensive vs. offensive planning}
In Fig.~\ref{fig:eval_compare:cummulative}, we plot the cumulative number of flows rejected by the planner over 14 rounds.
By pinning every active flow (cf. w/o reconfiguration, permanently pinning 100\% of flows in Fig.~\ref{fig:eval_compare:cummulative}), our approach performs defensive traffic planning, which results in a total of 62.5 rejected new flows on average.
In comparison, if the planner performs offensive traffic planning less than half as many flows---30.4 flows on average---are rejected (cf. w/ reconfiguration, permanently pinning 0\% of flows in Fig.~\ref{fig:eval_compare:cummulative}).

\subsubsection{Comparison with ILP}
Since the GFH is a heuristic, the question arises how \enquote{optimal} the GFH results are.
To answer this, we compare the GFH against a drop-in integer linear programming implementation, which is given the same input as the GFH, namely, the conflict graph \(\mathbf{G}(p')\) and the flow sets ActiveF(\(p\)) and ReqF(\(p\)), and has to compute \(p'\) which optimizes Eq.~\ref{eq:objective_weighted_max_color}.
To make it clear, this ILP does not encode the whole traffic planning problem itself, but requires the planner to encode the traffic planning constraints in \(\mathbf{G}(p')\).
The ILP constraints are given in Eq.~\ref{eq:ilp:obj}-\ref{eq:ilp:constraints_end}, and, except for the objective, correspond to the ILP used as \enquote{sure} algorithm in~\cite{falkTimeTriggeredTrafficPlanning2020}.
We have two sets of binary decision variables:
For each flow, \(\mathcal{X}_f[f]\in\{0,1\}\) encodes whether flow \(f\) can be admitted.
Similarly, a value of 1 for \(\mathcal{X}_c[v]\in\{0,1\}\) encodes that configuration \(v\) is added \(\mathcal{C}\).

{\small
\begin{align}
\max \sum_{\mathclap{f \in \operatorname{ActiveF}(p)}} \mathcal{X}_f[f] + \sum_{f \in \operatorname{ReqF}(p)} \frac{\mathcal{X}_f[f]}{\left|\operatorname{ActiveF}(p) \cup \operatorname{ReqF}(p)\right|} \label{eq:ilp:obj}\\
\text{subject to:} \qquad \forall (v, u) \in \mathcal{E}: \mathcal{X}_c[v] + \mathcal{X}_c[u] \leq 1\\
\forall f \in \operatorname{ActiveF}(p) \cup \operatorname{ReqF}(p): X_f[f] \leq \sum_{\mathclap{v \in \operatorname{cand}(f,p)}} \mathcal{X}_c[v] \label{eq:ilp:constraints_end}
\end{align}
}%

For this comparison, we saved the conflict graph instances from the last four planning rounds from Fig.~\ref{fig:eval_compare:cummulative}---presumably the largest conflict graphs in each scenario---to disk, and solved the corresponding ILP instances with a tool-chain using Julia~\cite{bezansonJulia2017}, JuMP~\cite{DunningHuchetteLubin2017} and Gurobi v9.1.1~\cite{gurobi}.
The ILP solver had a runtime limit of \SI{5}{\minute} (for comparison, the max. GFH runtime was \SI{8.3}{\second}).
Due to the higher memory requirements (up to \SI{56.694}{\giga\byte}) of the ILP tool chain, we used a server-grade computer with two AMD EPYC 7401 processors (each 24 cores) and \SI{256}{\giga\byte} RAM, but both, GFH and ILP solver, were limited to using max. 16 threads as in the other evaluations.

Fig.~\ref{fig:eval_compare:delta} plots for each planning round the relative difference in the number of rejected new flows \(\Delta \text{rejects} = \frac{\text{rejects GFH}-\text{rejects ILP} }{\text{new flows (total)}}\) over the number of candidate configurations in the conflict graph.
If the ILP solver rejected fewer flows, i.e., computed a better result we have \(\Delta \text{rejects}>0\), and vice-versa.
In Fig.~\ref{fig:eval_compare:delta}, the ILP solver could provide better solutions for small conflict graphs, and, except for a few outliers where GFH reverted to the result from the first phase, the advantage of the ILP over GFH is small.
Yet, once conflict graphs contain 30,000 or more candidate configurations, the ILP solver frequently hits the time-limit (for 0\% pinned flows: GFH avg=\SI{4.9}{\second}\(\pm\)\SI{1.8}{\second}; ILP avg=\SI{290.7}{\second}\(\pm\)\SI{71.4}{\second}) and would reject most new flows.
This means even if we factor in the performance benefits of the conflict graph modeling itself (cf.~\cite{falkTimeTriggeredTrafficPlanning2020}), the GFH algorithm pushes the boundaries for offensive traffic planning further out compared to state-of-the-art exact approaches (in this case, integer linear programming) which are limited to either small scenarios or a small fraction of reconfigurable flows (cf. pinning 80\% of all flows).

For reference, we also depict the average runtimes of the planner using the GFH for each round of the scenarios from Fig.~\ref{fig:eval_compare:delta} for the different ratios of permanently pinned flows (error-bars indicating the variance of the total runtime per round).
Figure~\ref{fig:eval_compare:runtimes} highlights the flexibility of our approach: if we \enquote{re-interpret} flow pinning as a tuning parameter, e.g., used for probabilistic pinning of new flows by the planner, we can cover the full range between offensive planning with high schedulability and extremely fast defensive planning (cf. Fig.~\ref{fig:eval_compare:runtimes}: max. \SI{250}{\milli\second} for processing rounds with \(\approx\)250 active flows on the server-grade machine).

\begin{figure*}[!htpb]
    \centering
    \begin{subfigure}[t]{0.23\linewidth}
        \centering
        \includegraphics[
        trim={.0cm .0cm .0cm .0cm},
        height=1.37in
        ]{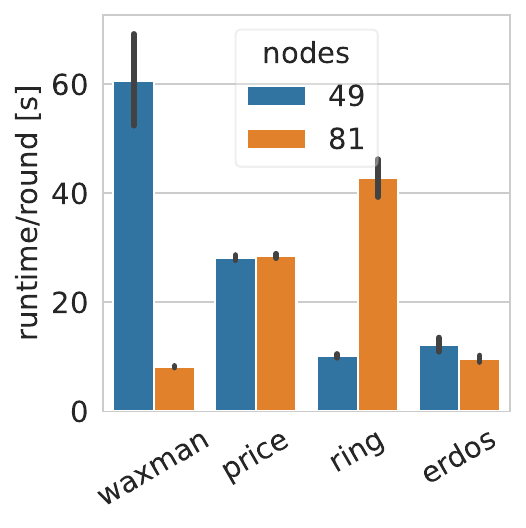}
        \caption{Runtime per round.}
        \label{fig:eval_topo:gfh}
    \end{subfigure}
    \hspace{1em}
    \begin{subfigure}[t]{0.34\linewidth}
        \centering
        \includegraphics[
        trim={.0cm .2cm .0cm .2cm},
        height=1.37in
        ]{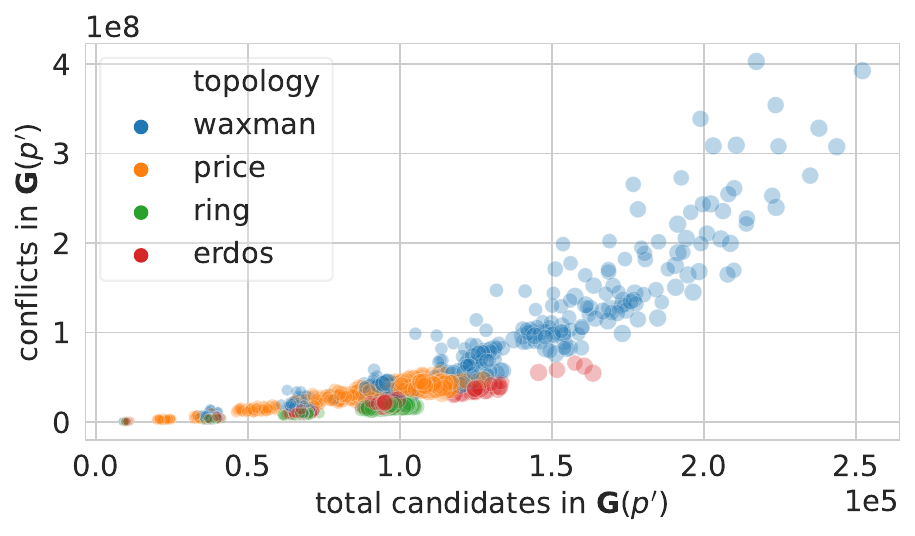}
        \caption{Conflict graph (49 nodes).}
        \label{fig:eval_topo:size49}	
    \end{subfigure}
    \hspace{1em}
    \begin{subfigure}[t]{0.34\linewidth}
        \centering
        \includegraphics[
        trim={.0cm .2cm .0cm .2cm},
        height=1.37in
        ]{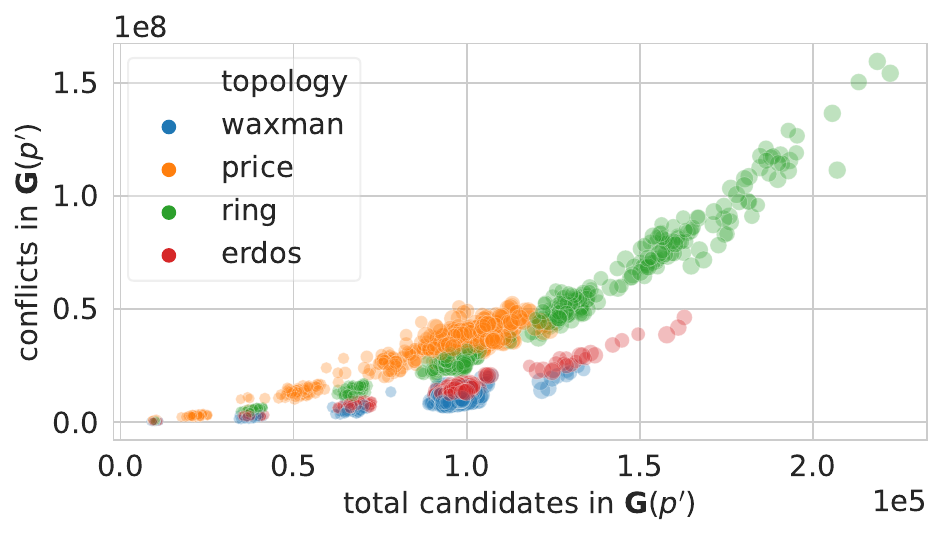}
        \caption{Conflict graph (81 nodes).}
        \label{fig:eval_topo:size81}	
    \end{subfigure}
    \caption{Impact of different network topologies on total runtime and conflict graph size.}
    \label{fig:eval_topologies}
\end{figure*}
\begin{figure*}[!htpb]
    \centering
    \begin{subfigure}[t]{0.23\linewidth}
        \centering
        \includegraphics[
        trim={.2cm .3cm .2cm .3cm},
        height=1in
        ]{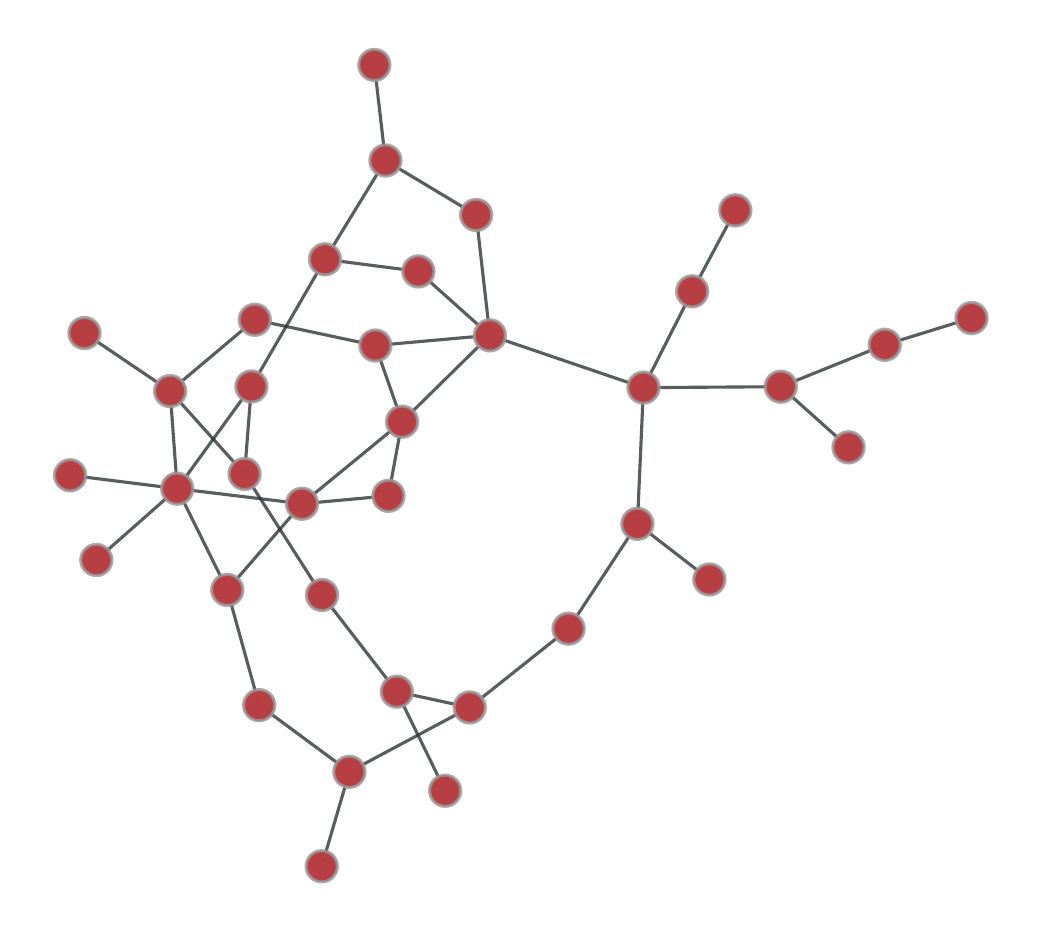}
        \caption{Waxman.}
        \label{fig:example_topologies:waxman}	
    \end{subfigure}
    \begin{subfigure}[t]{0.23\linewidth}
        \centering
        \includegraphics[
        trim={.2cm .3cm .2cm .3cm},
        height=1in
        ]{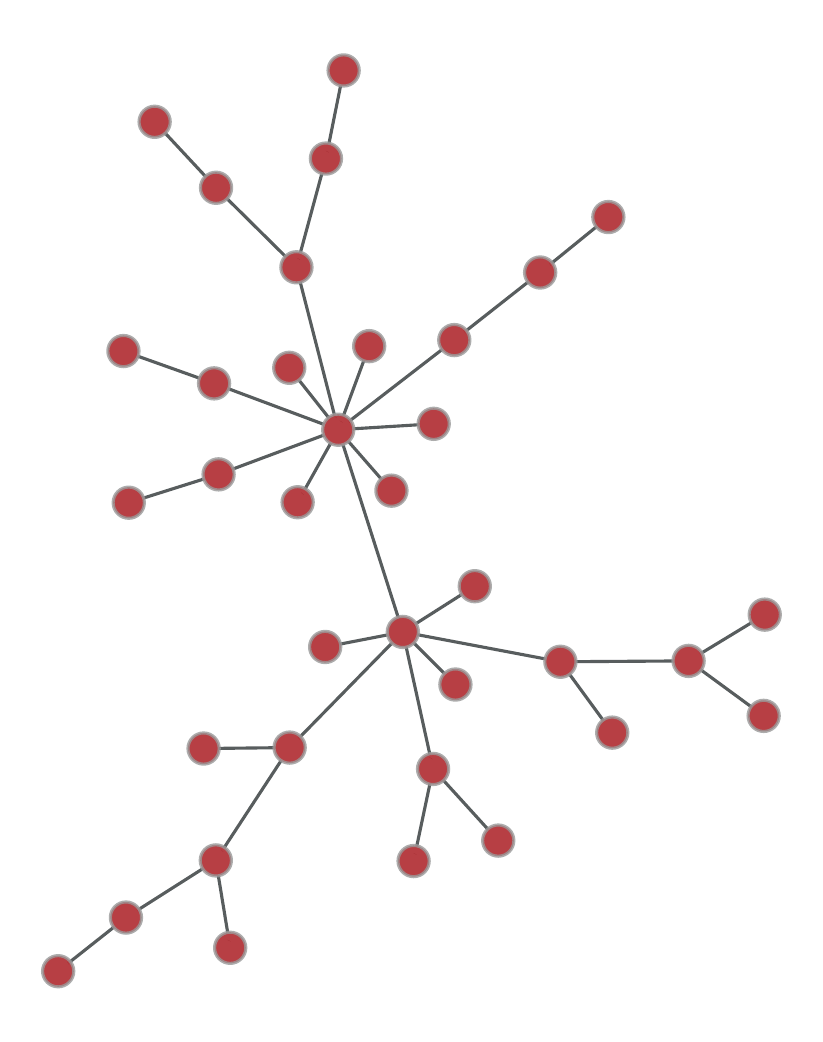}
        \caption{Price.}
        \label{fig:example_topologies:price}	
    \end{subfigure}
    \begin{subfigure}[t]{0.23\linewidth}
        \centering
        \includegraphics[
        trim={.2cm .3cm .2cm .3cm},
        height=1in
        ]{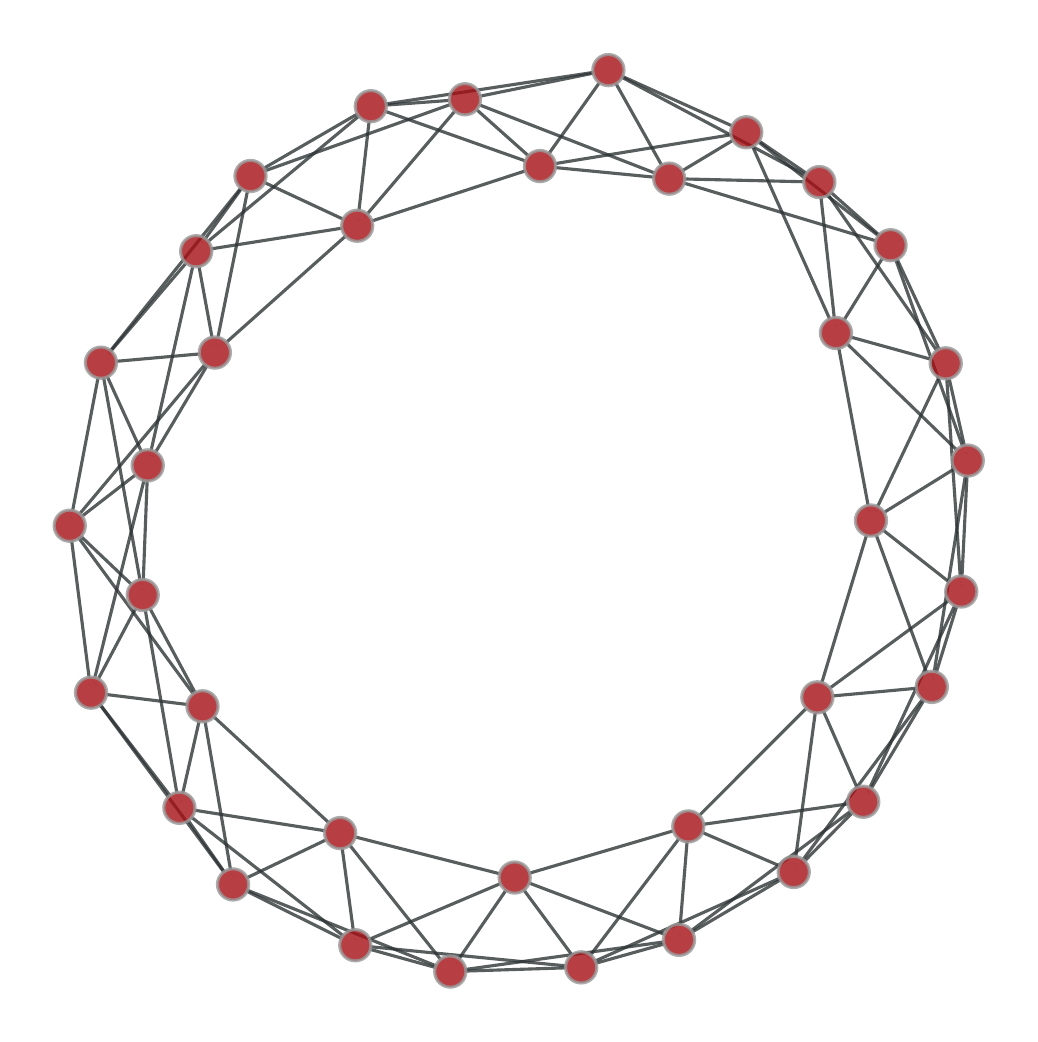}
        \caption{Ring(36,3).}
        \label{fig:example_topologies:ring}	
    \end{subfigure}
    \begin{subfigure}[t]{0.23\linewidth}
        \centering
        \includegraphics[
        trim={.2cm .3cm .2cm .3cm},
        height=1in
        ]{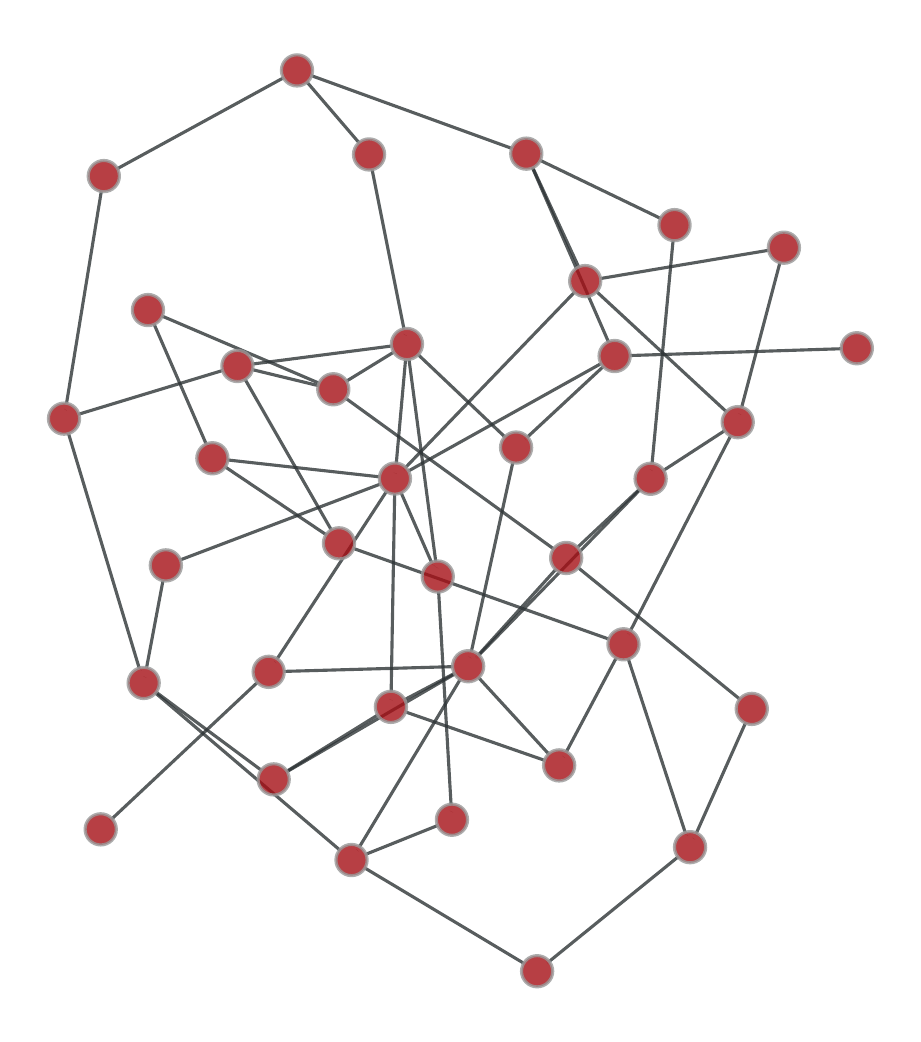}
        \caption{Erd\H{o}s-R\'enyi.}
        \label{fig:example_topologies:erdos}
    \end{subfigure}
    \caption{Example graph topologies for the different graph models with 36 nodes each.}
    \label{fig:example_topologies}
\end{figure*}

\subsubsection{Comparison with randomized heuristic}
\markChanges{
We also compare the GFH against another heuristic.
To the best of our knowledge, the \enquote{closest} competitor is the adapted version of Luby's maximal independent vertex set algorithm~\cite{luby1986simple}, which has been presented in our prior work \cite{falkTimeTriggeredTrafficPlanning2020}.
The adapted Luby's algorithm (aLA) is a randomized, iterative heuristic for the computation of \(\mathcal{C}\).
In each iteration of the aLA, eligible configurations are probabilistically added to \(\mathcal{C}\).
How likely it is for an eligible configuration \(c\) to be added to \(\mathcal{C}\) depends on the number of conflicts of \(c\) (less conflicts increases the probability), and the number of configurations already in \(\mathcal{C}\) that belong to the same flow \(f\) as \(c\) are already included in \(\mathcal{C}\) (\(f\) not admitted increases the probability).
The latter property adds some level of flow-awareness which the original Luby's algorithm lacks.
We additionally modify the aLA heuristic to account for the different importance of flows in ActiveF(\(p\)) and ReqF(\(p\)) (cf. Eq.~\ref{eq:objective_weighted_max_color}):
To ensure that no configuration of ReqF(\(p\)) shadows a configuration of a flow in ActiveF(\(p\)), we---similar to the GFH---first search configurations for active flows in the subgraph of \(\mathbf{G}(p')\) induced by all the candidates of ActiveF(\(p\)), i.e., the aLA \enquote{sees} only configurations and conflicts for flows in ActiveF(\(p\)), before we search for configurations for ReqF(\(p\)) in the \enquote{full} \(\mathbf{G}(p')\).
In other words, in the second phase of the heuristic the aLA \enquote{sees} configurations for all flows but can only select eligible configurations for flows in ReqF(\(p\)) that do not conflict with the configurations found during the first phase.
}

\markChanges{
We implemented the aLA heuristic in Julia and evaluated it with same methodology introduced for the comparison of GFH and ILP.
That is, we save the conflict graph instances to disk and use them as input for the aLA heuristic to compute \(p'\) and then compare the quality of the aLA heuristic and GFH in terms of admitted flows.
The results for the relative number of rejected flows \(\Delta \text{rejects} = \frac{\text{rejects GFH}-\text{rejects aLA} }{\text{new flows (total)}}\) are plotted in Fig.~\ref{fig:eval_compare:delta_ala}.
We see that the aLA heuristic never outperforms the GFH, i.e., different to the ILP, the aLA heuristic was never able to admit more flows than the GFH.
In absolute numbers, the aLA heuristic on average rejects 10.3 \emph{more} than the GFH (rejected flows: GFH avg=\(3.7\pm4.1\); aLA avg=\(13.9\pm9.4\))---given a total of 25 new flows in ReqF(\(p\)) .
While for ties for smaller \(\mathbf{G}(p')\), we predominantly observed that both, GFH and the aLA heuristic, admit similarly many (\(\approx 20\)) new flows, we observed that ties for larger \(\mathbf{G}(p')\) predominantly occur if both heuristics cannot admit any new flow, i.e., they have similarly \enquote{bad} results.
Remarkable is also the steep drop-off in the relative performance of the aLA once \(\mathbf{G}(p')\) contains several thousand candidate configurations (\(\Delta\)rejects: <5000 candidates avg=-0.10; \(\geq\)5000 candidates avg=-0.85).
}
\markChanges{
Considering that our aLA heuristic was implemented as a single-threaded Julia implementation, it is not surpising that it resulted in higher runtimes compared to the GFH (for 0\% pinned flows: GFH avg=\SI{4.9}{\second}\(\pm\)\SI{1.6}{\second}; aLA heuristic avg=\SI{29.2}{\second}\(\pm\)\SI{9.3}{\second}).
However, while it is reasonable to expect that---being much simpler---a multi-threaded aLA heuristic will match or even undercut the GFH runtime, this is payed for with much worse quality, especially for larger conflict graphs.
}
\subsubsection{Summary}
\markChanges{
Firstly, we showed in our evaluations that reconfiguring active flows during a traffic plan update reduces the number of rejected new flows.
Secondly, our evaluations highlight that the configuration selection strategy of the GFH is worth its complexity.
The GFH mostly stayed within close distance to the quality of the exact ILP solutions for smaller conflict graphs, and the GFH continued to provide good solutions within seconds even for those \enquote{larger} conflict graphs where the ILP ran into the runtime-limit.
On the other hand, in our evaluations, the simpler aLA heuristic at best managed to tie with the GFH, and on average the results of the aLA heuristic could admit only half the number of flows admitted by the GFH.
}

\subsection{Network}
\label{sec:eval:topology}
To investigate the effects of network topology and network size, we generated scenarios for networks with 49 nodes and 81 nodes for different graph models, namely, Waxman, Price, ring(\(n,k\)), and Erd\H{o}s-R\'enyi, using graph generators from~\cite{SciPyProceedings_11,peixoto_graph-tool_2014}.
Examples for the different network topologies are provided in Fig.~\ref{fig:example_topologies}. 
These networks cover a variety of different network characteristics, like diameter, number of alternative paths and scale-free property.
Thus, we can evaluate the fitness of our approach for different network types.
We evaluated 60 scenarios for each size-topology combination with 500 active flows on average.
After the initialization, the planner tries adding 50 new flows and removing 50 active flows in each planning round.

The average runtimes after initialization and warm-up rounds are plotted in Fig.~\ref{fig:eval_topo:gfh}.
We observe that a larger network does not per-se result in higher runtimes or bigger conflict graphs, cf. Fig~\ref{fig:eval_topo:size49}, \ref{fig:eval_topo:size81}.
Waxman networks with 49 nodes result in the largest conflict graph sizes overall and highest average (\SI{60.5}{\second}) and maximal (\SI{283.72}{\second}) total runtime per round.
In Price networks with only one path between every two nodes, the planner cannot resolve conflicts via routing.
Compared to the other topologies, on average more than ten times as many new flows were rejected per round (49n. avg=10.4\(\pm\)5.12 rej./round; 81n. avg=11.6\(\pm\)5.19 rej./round) in scenarios with Price network topologies.
In scenarios with Price networks, GFH often required all re-runs, and we measured the third-highest total runtimes.
This suggests that our approach scales with the actual difficulty rather than just the network size.

\section{Conclusion and Outlook}
\label{sec:conclusion_outlook}
In this paper, we presented a novel approach for offensive dynamic traffic planning, which allows for the reconfiguration of active flows to achieve better network utilization, and a novel algorithm for computing these traffic plans.
We are able to quantify and control the QoS degradation a flow may suffer during such reconfigurations.
Our evaluations show that we can efficiently compute updated traffic plans for scenarios with hundreds of active flows.

Interesting directions for future work include investigating memory-access optimized conflict graph data-structures and the relaxation of the zero-queuing constraint, as well as protocols and concepts for error-handling while computing the traffic plan and during traffic plan deployment.

\ifCLASSOPTIONcaptionsoff
  \newpage
\fi



%

\bibliographystyle{IEEEtran}
\bibliography{IEEEabrv,references}

\begin{thebibliography}{10}
\bibitem{noauthor_hello_2016}
KUKA, ``Hello industrie 4.0 \_we go digital.''
\end{thebibliography}

%


\begin{IEEEbiography}[{\includegraphics[width=1in,height=1.25in,clip,keepaspectratio]{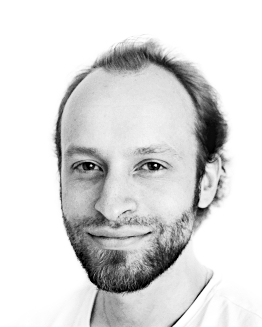}}]{Jonathan Falk}
	joined the Distributed Systems group of the Institute for Parallel and Distributed Systems of the University of Stuttgart after receiving his Master’s degree in Electrical Engineering and Information Technology. His research interests include time-sensitive communication, time-triggered networks, and cyber-physical systems.
\end{IEEEbiography}

\begin{IEEEbiography}[{\includegraphics[width=1in,height=1.25in,clip,keepaspectratio]{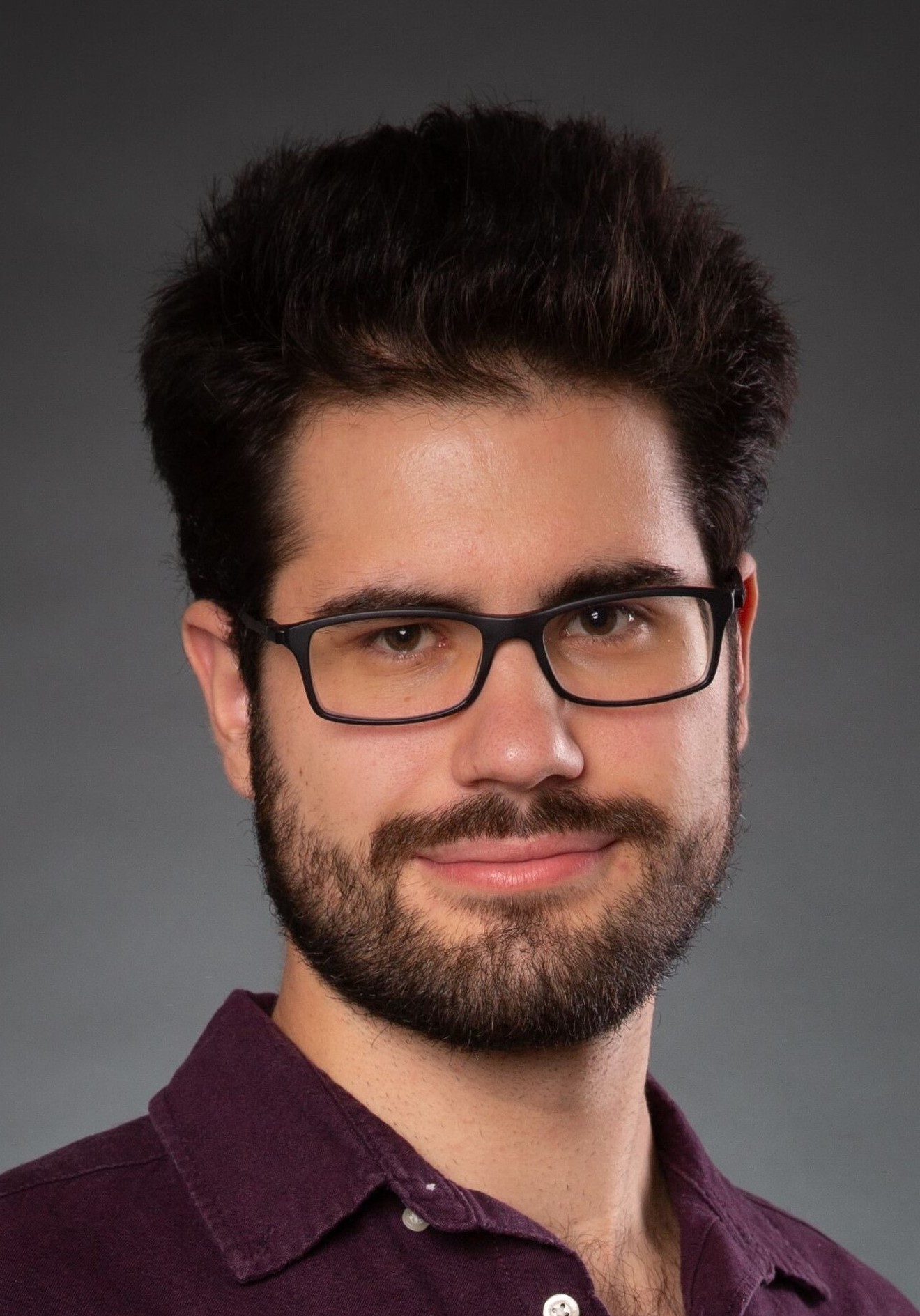}}]{Heiko Geppert}
	received his Master degree in Software Engineering from the University of Stuttgart, Germany, in 2019.
	He is currently working on his PhD at the Distributed Systems group of the University of Stuttgart.
	His research interests include distributed graph processing and network scheduling, with a focus on fast and scalable solutions.
\end{IEEEbiography}

\begin{IEEEbiography}[{\includegraphics[width=1in,height=1.25in,clip,keepaspectratio, trim={0.225cm 0 0.225cm 0}, clip]{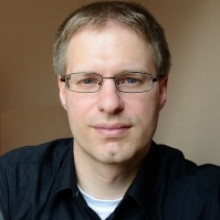}}]{Frank Dürr}
	is a senior researcher and lecturer at the Distributed Systems Department of the Institute of Parallel and Distributed Systems (IPVS) at University of Stuttgart, Germany. He received both, his doctoral degree and diploma in computer science from University of Stuttgart. Frank Dürr is currently leading the mobile computing and the software-defined networking (SDN) \& time-sensitive networking (TSN) groups of the Distributed Systems Department. Besides several scientific papers in the area of SDN and TSN, including for instance concepts for scheduling and routing in TSN/SDN networks, he has given tutorials on SDN at several conferences. Besides SDN and TSN, his research interests include mobile and pervasive computing, location privacy, and cloud computing aspects overlapping with these topics like mobile cloud and edge computing, or data center networks.
\end{IEEEbiography}

\vfill

\begin{IEEEbiography}[{\includegraphics[width=1in,height=1.25in,clip,keepaspectratio]{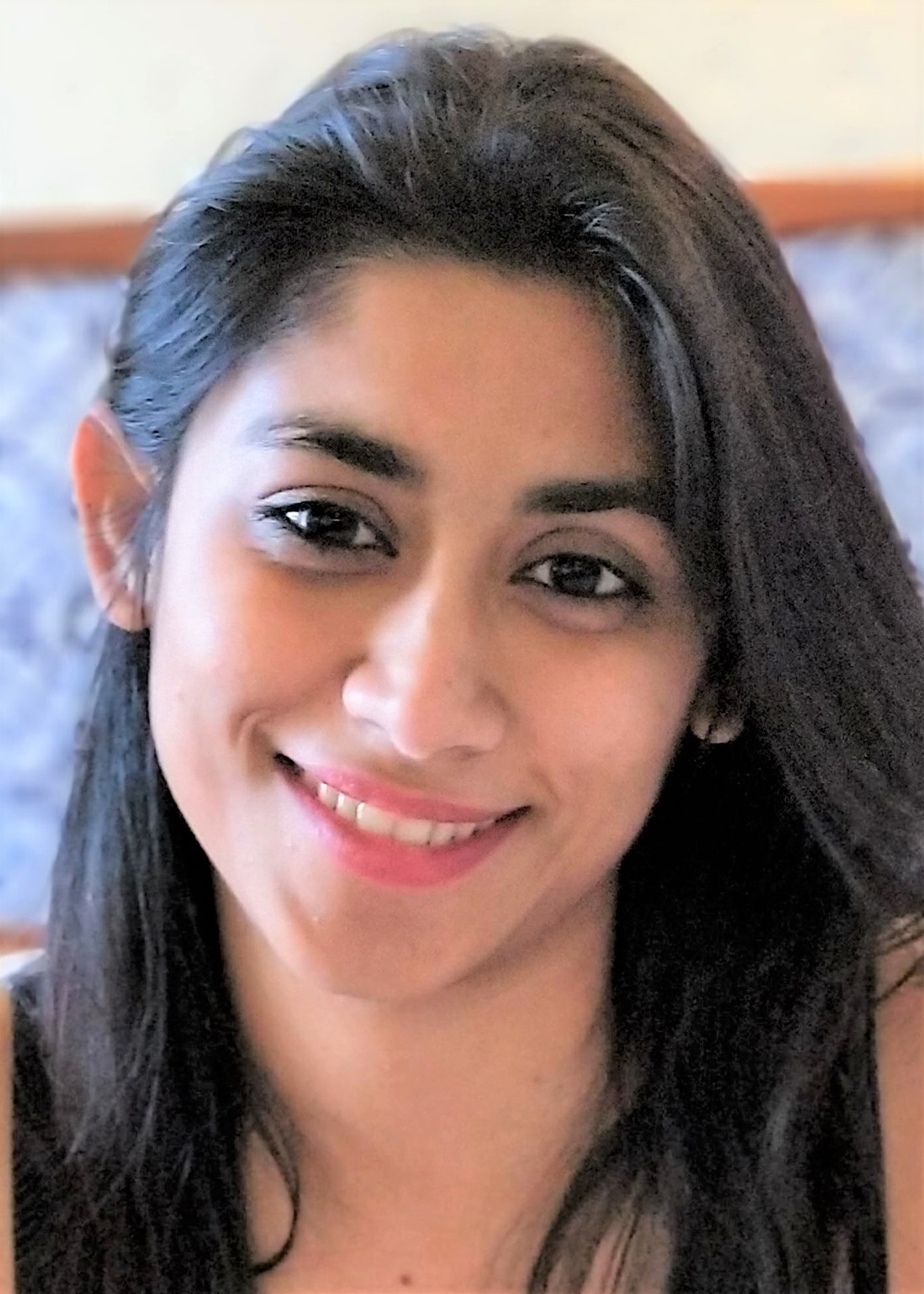}}]{Sukanya Bhowmik}
	received her doctoral degree from University of Stuttgart, Germany, in 2017. She is currently working as a postdoctoral researcher at the Distributed Systems research group of University of Stuttgart. Her research interests include stream/complex event processing, high performance communication middleware, in-network event processing, software-defined networking, and distributed graph processing,  with a focus on scalability, line-rate performance, resource efficiency, and adaptability aspects.
\end{IEEEbiography}

\begin{IEEEbiography}[{\includegraphics[width=1in,height=1.25in,clip,keepaspectratio]{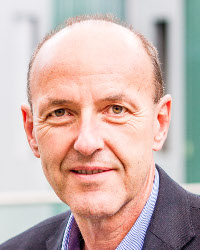}}]{Kurt Rothermel}
	received his doctoral degree in Computer Science from University of Stuttgart in 1985. From 1986 to 1987 he was a Post-Doctoral Fellow at IBM Almaden Research Center in San Jos\'e, U.S.A., and then joined IBM's European Networking Center in Heidelberg. In 1990,  he became a Professor for Computer Science at the University of Stuttgart. From 2003 to 2011 he was head of the Collaborative Research Center Nexus (SFB 627), conducting research in the area of mobile context-aware systems. His research interests are in the field of distributed systems, computer networks, and mobile systems.computer networks, and mobile systems.             
\end{IEEEbiography}

\vfill


\end{document}